\documentclass[
    pdflatex,
    sn-mathphys,
    dvipsnames 
]{sn-jnl}

\usepackage{tabularx}
\usepackage{multicol}
\usepackage{multirow}
\usepackage{colortbl} 
\usepackage{array}
\usepackage{booktabs}
\usepackage[utf8]{inputenc}	
\usepackage[T1]{fontenc}	
\usepackage{xcolor}
\usepackage{rotating}
\usepackage{nicefrac}
\usepackage{cleveref}

\usepackage[subrefformat=simple,labelformat=simple]{subcaption}
\makeatletter
\algnewcommand{\StateXfirst}{\Statex \hskip\ALG@tlm} 
\algnewcommand{\StateX}{\Statex \hskip\ALG@tlm}
\algnewcommand{\StateXmulti}[1]{\Statex \hskip\ALG@tlm%
  \parbox[t]{\dimexpr\linewidth-\ALG@thistlm}{\strut#1\strut}}
\makeatother

\algrenewcommand\algorithmicwhile{\textbf{While}}
\algrenewcommand\algorithmicif{\textbf{If}}

\algtext*{EndWhile}
\algtext*{EndIf}

\let\oldState\State
\newcommand*{\stopnumbering}{%
  \let\olditem\item
  \renewcommand{\item}[1][]{\olditem[]}%
  \let\State\Statex}
\newcommand*{\resumenumbering}{%
  \let\item\olditem
  \let\State\oldState}

\hypersetup{hypertexnames=false}

\algrenewcommand\alglinenumber[1]{\footnotesize #1:}
\renewcommand{\rmdefault}{cmr}

\newcommand\N{\mathbb{N}}

\usepackage{bbm}

\newcommand*\di{\mathop{}\!\mathrm{d}}

\DeclareMathOperator{\Exp}{Exp} 

\renewcommand\epsilon{\varepsilon}

\renewcommand\theta{\vartheta}

\newcommand*{\num}[1]{{#1}}

\newcommand*{\nS}{\num{S}}
\newcommand*{\nI}{\num{I}}
\newcommand*{\nR}{\num{R}}
\newcommand*{\nD}{\num{D}}

\graphicspath{{./../}}

\definecolor{ballblue}{rgb}{0.1, 0.67, 0.7}
\definecolor{tabblue}{HTML}{1f77b4}
\definecolor{taborange}{HTML}{ff7f0e}
\definecolor{tabgreen}{HTML}{2ca02c}
\definecolor{tabred}{HTML}{d62728}
\definecolor{tabpurple}{HTML}{9467bd}
\definecolor{tabbrown}{HTML}{8c564b}
\definecolor{tabpink}{HTML}{e377c2}
\definecolor{tabgray}{HTML}{7f7f7f}
\definecolor{tabolive}{HTML}{bcbd22}
\definecolor{tabcyan}{HTML}{17becf}

\jyear{2022}%

\theoremstyle{thmstyleone}%
\newtheorem{theorem}{Theorem}
\newtheorem{lemma}[theorem]{Lemma}%

\theoremstyle{thmstyletwo}%
\newtheorem{remark}{Remark}%

\theoremstyle{thmstylethree}%
\newtheorem{assumption}{Assumption}%

\begin{document}

\title[Generalised Gillespie Algorithms]{Generalised Gillespie Algorithms for Simulations in a Rule-Based Epidemiological Model Framework}


\author[1]{\fnm{David} \sur{Alonso}}\email{dalonso@ceab.csic.es}
\equalcont{These authors contributed equally to this work.}

\author[2]{\fnm{Steffen} \sur{Bauer}}\email{steffen.bauer@stud.uni-heidelberg.de}
\equalcont{These authors contributed equally to this work.}

\author*[3]{\fnm{Markus} \sur{Kirkilionis}}\email{mak@maths.warwick.ac.uk}
\equalcont{These authors contributed equally to this work.}

\author[4]{\fnm{\\Lisa Maria} \sur{Kreusser}}\email{lmk54@bath.ac.uk}
\equalcont{These authors contributed equally to this work.}

\author[5]{\fnm{Luca} \sur{Sbano}}\email{luca.sbano@posta.istrizione.it}
\equalcont{These authors contributed equally to this work.}


\affil[1]{\orgdiv{Theoretical and Computational Ecology}, \orgname{Center for Advanced Studies of Blanes (CEAB-CSIC), Spanish Council for Scientific Research}, \orgaddress{\street{Acces Cala St. Francesc 14}, \city{Blanes}, \postcode{E-17300}, \country{Spain}}}

\affil[2]{\orgdiv{Mathematisches Institut, Mathematikon}, \orgname{Ruprecht-Karls-Universit\"at Heidelberg}, \orgaddress{\street{Im Neuenheimer Feld 205}, \postcode{69120} \city{Heidelberg}, \country{Germany}}}

\affil*[3]{\orgdiv{Warwick Mathematics Institute}, \orgname{University of Warwick}, \orgaddress{\street{Zeeman Building}, \city{Coventry}, \postcode{CV4 7AL}, \country{United Kingdom}}}

\affil[4]{\orgdiv{Department of Mathematical Sciences}, \orgname{University of Bath}, \orgaddress{\city{Bath}, \postcode{BA2 7AY}, \country{United Kingdom}}}

\affil[5]{\orgname{Liceo “Vittoria Colonna”}, \orgaddress{\street{Via dell’Arco del Monte 99}, \city{Rome}, \postcode{00186}, \country{Italy}}}


\abstract{Rule-based  models have been successfully used to represent different aspects of the COVID-19 pandemic, including age, testing, hospitalisation, lockdowns, immunity, infectivity, behaviour, mobility and vaccination of individuals. These rule-based approaches are motivated by chemical reaction rules which are traditionally solved numerically with the standard Gillespie algorithm proposed in the context of molecular dynamics. When applying reaction system type of approaches to epidemiology, generalisations of the Gillespie algorithm are required due to the time-dependency of the problems. In this article, we present different generalisations of the standard Gillespie algorithm which address discrete subtypes (e.g., incorporating the age structure of the population), time-discrete updates (e.g., incorporating daily imposed change of rates for lockdowns) and deterministic delays (e.g., given waiting time until a specific change in types such as release from isolation occurs). These algorithms are complemented by relevant examples in the context of the COVID-19 pandemic and numerical results.}

\keywords{Gillespie algorithm, rule-based models, COVID-19.}



\maketitle

\section{Introduction}
\label{sec:intro}

This paper discusses generalisations of the Gillespie algorithm, an algorithm which was originally designed for chemical reaction systems. The generalisations are triggered by the wish to apply a reaction system type of approach to epidemiology, a choice that is often done, even unconsciously, due to the simple analogy of the infection process with a chemical reaction. However, the important reason why generalisations of the Gillespie algorithm are necessary is already deeply rooted in the history of mathematical modelling, where time-dependent problems are, in a majority of cases, modelled by ordinary differential equations (ODE), and in that class of models, by so-called mass-action systems. Roughly speaking, a mass-action system can be characterised by the rate of change of concentrations of an entity being proportional to the concentration of entities needed to assemble or change it during collision events. Using mass-action systems in modelling is very helpful, as the whole arsenal of methods developed for dynamical systems can be applied. In addition, the right-hand side of a mass-action based ODE is polynomial, allowing the additional use of specific mathematical theories, like algebraic geometry. However, there are a few assumptions needed for mass-action systems to work:

\begin{description}
\item[(i)] The system is deterministic, which is usually the case when the colliding entities are present in very large numbers. In addition, in the dynamical systems sense, the system is autonomous.
\item[(ii)] The collision events between entities are such that every entity can collide with any other entity in the system, and has moreover exactly the same probability to collide with another entity in the system.
\item[(iii)] The system is confined to a well-defined (thermodynamically) closed space, in reaction systems called the reaction chamber. This also constitutes a well-defined system boundary, where 'inside' the system and 'outside' the system are well-defined concepts. Often such systems are also labeled as 'closed' and 'open' if either no disturbances from outside the system boundaries are considered in a closed system, or a allowed in an open system.
\item[(iv)] By the collision assumption, system changes are driven by a sequence of events ordered by time. Moreover, system changes are instantaneous, and only depend on the current state of the system when the event happens. This induces the so-called Markov property.
\end{description}

The original Gillespie algorithm is constructed to generalise assumption $(i)$, deterministic dynamics. The idea is to relax the idea of very large number of entities in the system, and therefore to allow for so-called small-number effects. This transforms the ODE mass-action system into a stochastic process. We will follow this argument and assume assumption $(i)$ is not necessarily valid, as we believe, especially for epidemiological processes, we need to study the consequences of stochastic effects. However, in this article, we keep assumptions $(ii)$ and $(iii)$, as these assumptions are also fundamental in order to apply any Gillespie-type algorithm. However, we will need to relax assumption $(iv)$, because epidemiology (but also other applications) provides clearly model cases where an event triggers a system change only with a delay. Vaccination is such an example. The vaccination event is instantaneous, however the establishment if immunity, for example, will only happen after some time has passed. The algorithms presented in this paper should be able to enable rule-based models to be used in situations where simple reaction mechanisms fail. But what are rule-based systems?

\subsection{Rule-based systems}
\label{sec:intro:subsec:rules}

Rule-based systems are a generalisation of reaction systems, mostly in order to carry over the theory to other fields of applications. Following the notation in \cite{rbse} for the rule-based modelling framework, we denote by $\mathcal{T}=\{T_1,\ldots, T_s\}$, $s \in \mathbb{N}$, the set of all possible types or individual characteristics, $s$  different types in total. 
 Let $\mathcal{R}=\{R_1,\ldots, R_r\}$, $r \in \mathbb{N}$,  be the finite set of rules constituting the epidemiological system.  Each rule $R_j$ takes the form
\begin{eqnarray}
\label{eq:sec:model:subsec:reaction:reactionscheme}
\sum_{i=1}^{s} \alpha_{ij}  T_i  \xrightarrow{k_j} \sum_{i=1}^{s} \beta_{ij}  T_i,
\end{eqnarray}
for $j=1,\ldots,r$, where $k_j$ denotes the rate constant of reaction $R_j$. 
The coefficients $\alpha_{ij} \in \mathbb{N}_0$ are the stoichiometric coefficients of types $T_i$, $1 \le i \le s$, of the {\bf source} side, and  $\beta_{ij}\in \mathbb{N}_0$ are the stoichiometric coefficients of types $T_i$ of the {\bf target} side of  transition or reaction $R_j$. Note that we use the terms \emph{reaction} and \emph{rule of interaction} interchangeably for $R_j$.

The rule-based formulation \eqref{eq:sec:model:subsec:reaction:reactionscheme} of reactions can be  interpreted   stochastically or deterministically. While the deterministic formulation is based on a system of ODEs, several stochastic formulations exist \cite{rbse}, including continuous-time Markov chains, Kolmogorov Differential Equations and Master Equations, as well as on Stochastic Differential Equations. Here we focus on stochastic dynamics based on  continuous-time Markov chains in this paper. We  introduce the population size $N$ and consider the size $n_i(t)$ of type $T_i$ at time $t\geq 0$ for $i=1,\ldots,s$ as discrete random variables with values in $\{0,\ldots,N\}$. For ease of notation, we also write $n=(n_1,\ldots,n_s)$, which accurately characterizes the system configuration (or state of the system) at any given time. Further formulations include Kolmogorov Differential Equations and Master Equations, as well as on Stochastic Differential Equations, see \cite{rbse}.

The time-dependency of  $n$ motivates to regard $n$ as a stochastic process and under some additional assumptions, $n$ is a continuous-time Markov chain. To simulate the sample paths of $n$, one can apply the Gillespie algorithm. The simulation of many trajectories of the system  allows the computation of the statistics of the evolution. While it is known that polynomial regression asymptotically converges to the numerical solution as the number of considered trajectories increases, there exist realisations so that the stochastic trajectories and the deterministic dynamics differ significantly. We believe we need this multi-sampling approach in order to make better predictions on epidemic scenarios when compared with deterministic models.

\subsection{The pandemic as a modelling challenge}
\label{sec:intro:subsec:pandemic}

The Coronavirus disease 2019 (COVID-19) pandemic has put epidemic modelling in the worldwide spotlight. However, public policy making and forecasting the spread of COVID-19  remains a challenging task \cite{Manrubia2020,Castro26190}. This is especially true because the huge impact of the pandemic had many previously unknown or unconsidered consequences that came to the spotlight and which need to be reflected in models in order to make them predictive. Moreover the evolution and unfolding of the pandemic triggered new policies and interventions which need unconventional mathematical modelling techniques, including novel simulation techniques. This latter need essentially triggered the writing of this paper.  

In terms of mathematical models supporting policy decisions a lot of trade-offs have been registered during the Covid-19 pandemic. The most obvious one is the link of a pandemic to economy, where of course the pandemic policy measures have created huge impacts on economic performance, which in turn had tremendous impact on the spread of the disease. An example of an attempt to model such feedback for the Philippines is given in \cite{Lara-Tuprio:2022aa}, but this model is based on a pure differential equation approach, as most models in this area. Of course the pandemic has also positive economic aspects, for example in terms of air pollution \cite{Liu:2022aa}. If we like to give good estimates of such models in terms of stochastic modelling, algorithm like the one developed in this paper will be needed.

There was a multitude of lockdown measures implemented during the Covid-19 pandemic, and it is still an open question which one worked best, and which one could be further improved in their effectiveness. The most severe intervention was surely a country-wide lockdown with a simultaneous closing of national borders. These measures prove effective \cite{Wang:2022aa}, but costly in economic terms \cite{Shan:2021aa}. Then there are local lockdowns, for example automatically triggered by infection numbers, as based on testing \cite{Wells:2022aa}. Isolation measures were also implemented on an individual level, after a person had been detected to be infected. Here the question is how effective the search for infected individuals must be (contact tracing \cite{Aleta:2020aa}, how long such an isolation should last, and how much impact isolations have on economy and society \cite{Eilersen:2020aa}? More mild measures are to impose protective measures, like face-mask wearing, or to avoid close contact \cite{Block:2020aa}. To study such effects, and eventually to optimise the effectiveness of measures one needs to use the stochastic algorithms developed in this paper, which can be applied to such non-autonomous stochastic systems describing interventions.

Many such mathematical modelling problems involve some form of delay, constituting also a basic form of memory. This can perhaps be best explained in terms of vaccination strategies. Many vaccination strategies involve two of more doses of vaccine injections after some fixed time delay, which improves the effect of the vaccine drastically \cite{Hall:2022aa}, sometimes better than immunity by infection \cite{Waxman:2022aa}. Immunity itself also is never instantaneously achieved after an infection or vaccination event, but build up and declines over a longer period of time. If such effects or policies need to be incorporated into a stochastic collision-based stochastic model, then again the algorithms in this paper need to be used.

\subsection{Gillespie algorithm}
\label{sec:intro:subsec:gillespie}

In the context of molecular dynamics, Gillespie  showed in a series of seminal papers \cite{Gillespie1976,Gillespie1992,Gillespie2000} that, if the reaction rates are proportional to the product of the number of particles, i.e.\ if they follow the \emph{mass action} kinetics, then the evolution processes are Poisson processes and can be reconstructed through an algorithm. Note that \emph{mass action} is equivalent to random \emph{close contacts}. The Gillespie algorithm is characterised by observing that Poisson processes have inter-event times (\emph{waiting time}) which are exponentially distributed and satisfy the Markov property. This is equivalent to  increments being independent  and the time evolution having no memory.  The spread of an infection requires to analyse effects and interactions that occur at  different scales in time, space and size, which is a  challenging problem for modellers. Such effects can be more accurately modelled  through memory effects, and therefore requiring the introduction of non-Markov process with non-exponentially distributed waiting-times. The necessity of a generalisation of the Gillespie algorithm was considered already by Gillespie himself in \cite{Gillespie_ME_with_waiting_times}.

The main directions to generalise the Gillespie algorithm is to find ways to efficiently simulate systems with a large number of degrees of freedom \cite{Gibson2000,Bernstein2005}, and simulate stochastic processes with given waiting time distributions not necessarily exponentially distributed. These problems have  recently attracted lots of interest, including generalisations of the Gillespie algorithm for non-Markov processes  \cite{Masuda,Bog,CCTRW}, based on considering waiting times which are not necessarily exponentially distributed.  A recent general review can be found in \cite{Masuda2}. 

In this work, we propose extensions of the Gillespie algorithm for a rule-based modelling framework \cite{rbse} in the context of epidemiology.  
In these models, individuals are represented by discrete entities, and contagion processes are also discrete. Due to the complexities of these models, standard algorithms cannot be immediately applied which requires the  development of new algorithms for computing the model solutions numerically. Motivated by the approach in \cite{CCTRW}, we introduce several generalisations of the Gillespie algorithm in this paper.




\subsection{Outline}

This paper is organised as follows. Considering the general rule-based framework for  epidemiological models \cite{rbse}, we present different generalisations of the standard Gillespie algorithm in Section~\ref{sec:gillespie}, which  address  discrete subtypes (e.g., incorporating the age structure of the population), time-discrete updates (e.g., incorporating daily imposed change of rates for lockdowns) and deterministic delays (e.g., given waiting time until a specific change in types such as release from isolation occurs). 
Finally, in Section~\ref{sec:simulations}, we apply the different algorithms to several relevant examples and show numerical results. The examples also present some possible answers to the challenges outlined in section \ref{sec:intro:subsec:pandemic}.

\section{Gillespie algorithm and its generalisation}
\label{sec:gillespie}

In this section we describe the main algorithms used to study the stochastic evolution of rule-based models of the form \eqref{eq:sec:model:subsec:reaction:reactionscheme}, i.e., how the stochastic process makes $n$ evolves in time. The evolution arises due to the interactions among types $\mathcal{T}$ and, in the language of reaction schemes, can be expressed  as a reaction network. 

For defining the transition probabilities associated with the stochastic process we assume that the process is time-homogeneous, i.e., the transition probabilities depend on the time between events $\Delta t$, but not on the specific time $t$. In addition, we suppose that the Markov property is satisfied, that is, the probability distribution of a future state of the stochastic process at time $t+\Delta t$ only depends on the current state at time $t$, but not on any state prior to $t$. Given the system is in state $n$, we introduce the propensities $\kappa_j(n)$ for $j=1,\ldots,r$. Here, the propensity $\kappa_j(n)$ of the $j$th reaction $R_j$ is the probability that the $j$th reaction occurs within an infinitesimal time interval and is given by

\begin{align}\label{eq:propensity}
 \kappa_j(n) = k_j  \prod_{i=1}^s   \begin{pmatrix} n_i \\ \alpha_{ij} \end{pmatrix},
\end{align}
where $k_j$ denotes the rate constant of reaction $R_j$ and the combinatorial factor reflects the number of ways in which reaction $R_j$ may happen. Based on the assumption that at most one reaction can occur within an infinitesimal time interval $\Delta t$, the transition rates for the Markov process, given that the system is in state $n$, can then be approximated as $\kappa_j(n)\Delta t$ for reaction $R_j$ to occur within time interval $\Delta t$. 
As in the original paper \cite{Gillespie1976} by Gillespie, we assume in the following:

\begin{assumption}
\label{Gillespie_assumption}
Reaction $R_j$ occurs within a time interval $\Delta t$ with probability $\kappa_j(n)\Delta t$.
\end{assumption}

For a given population with system configuration $n$, one can  show that the probability density associated with the occurrence of $R_j$ is given by
%
\begin{equation}
\Phi_j(n,t)=\kappa_j(n)\exp\left(-\sum_{l=1}^r\kappa_l(n)\,t\right)
\label{eq:gillespie_prob_Pi}
\end{equation}
%
as a function of time $t$. The probability that reaction $R_j$ ever occurs is given by 
%
\begin{equation}
\rho_j(n)=\int_0^\infty dt'\,\Phi_j(n,t')=\frac{\kappa_j(n)}{\sum_{l=1}^r \kappa_l(n)}
\label{eq:rho_i_standardgil}
\end{equation}
%
for  $j\in\{1,...,r\}$.

The reactions in  \eqref{eq:sec:model:subsec:reaction:reactionscheme} generate a stochastic process which satisfies the Markov property due to Assumption~\ref{Gillespie_assumption}. This process is  the starting point for the derivation of the master equation (ME), describing the time evolution of the probability distribution of the system as a whole
(see \cite{Gardiner}). The probability density in \eqref{eq:gillespie_prob_Pi} can  be interpreted  in terms of reaction waiting times and, in fact, for
Assumption~\ref{Gillespie_assumption} we can say equivalently that  reaction $R_j$ occurs with an \emph{exponentially distributed} reaction waiting time.

In \cite{Gillespie1976} Gillespie showed that the ME can be exactly solved through the construction of realisations of the stochastic process associated to \eqref{eq:sec:model:subsec:reaction:reactionscheme}.    This is achieved through Algorithm \ref{alg:sec:model:subsec:gillespie}, usually called Direct Method, where $::=$ denotes replacement.

The random numbers generated in Algorithm~\ref{alg:sec:model:subsec:gillespie} are used for determining the index $i$ of the next reaction to occur and the time interval $\tau$ for reaction $R_i$ to occur.
By the definition of $\rho_l(n)$ in \eqref{eq:rho_i_standardgil} we have
$$\frac{1}{\Phi} \sum_{l=1}^i \kappa_l(n)=\sum_{l=1}^i \int_0^\infty dt'\,\Phi_l(n,t')=\sum_{l=1}^{i}\rho_l (n)$$ in step~\ref{alg:line:GillespieFour} in Algorithm~\ref{alg:sec:model:subsec:gillespie}.

We use the Gillespie algorithm to simulate the model dynamics, based on the assumption that the processes have exponentially distributed waiting times or, equivalently, that they are all Poisson processes. such as atomic or molecular mixtures of gases or liquids. In complex systems formed by types which are heterogeneous and evolve through diverse interactions that cannot  always be characterised by \emph{exponentially distributed waiting times}, the Gillespie approach is not fully adequate. 
Several generalisations for arbitrary distributions exist in the literature, see for instance \cite{Gillespie_ME_with_waiting_times,Masuda,Bog,CCTRW}.

\begin{algorithm}[htb] 
\caption{Gillespie algorithm}
\label{alg:sec:model:subsec:gillespie}
\begin{algorithmic}[1] \setcounter{ALG@line}{-1}
\State Initialize the starting time $t$, the system configuration $n=(n_1,\ldots,n_s)$ and the rate constants $k_j$ for $j=1,\ldots,r$. 
\State Generate two random numbers $r_1, r_2$ uniformly distributed in $[0,1]$. \label{alg:line:GillespieFirst}
\State Set $\Phi:= \sum_{l=1}^r \kappa_l (n)$ and compute $\tau:= -\frac{1}{\Phi} \ln r_1$. \label{alg:line:GillespieTwo}
\State Set $t::=t+\tau$ as the time of the next rule execution. \label{alg:line:GillespieThree}
\State Determine the reaction $R_i$ which is executed at time $t$ by finding $i$ such that \label{alg:line:GillespieFour}
\begin{align}
 \label{eq:conditiongillespie}
 \frac{1}{\Phi} \sum_{l=1}^{i-1}\kappa_l (n)<   r_2 \leq \frac{1}{\Phi} \sum_{l=1}^{i}\kappa_l (n).
\end{align}
\State Execute rule $R_i$ and update the new system configuration $n$. \label{alg:line:GillespieFive}
\State \textbf{If} $t<T$ go to step~\ref{alg:line:GillespieFirst}, \textbf{otherwise} stop.
\end{algorithmic}
\end{algorithm}

\begin{remark}[Type-dependent reaction rates]\label{rem:model:subsec:gillespie:subsub:typedependentrates}
Type-dependent reaction rates are an important way for  modelling the influence of different types on each other in the context of epidemics and it is of interest to include this aspect in the modeling framework. For example, the behaviour of individuals plays a crucial role and people's behaviour is often influenced by the behaviour of people around them. One way to include people's behaviour is to split the susceptible individuals, denoted by type $S$, into two subtypes: \emph{wise} ($W$) and \emph{risky} ($R$) individuals. While individuals of type $W$ generally abide by the  recommendations and regulations communicated by health authorities, individuals of type $R$ are not worried by the situation and, in general, want to carry on with their usual habits, see  \cite{rbse} for a detailed description of the model. To describe the reaction $W\xrightarrow{k_W} R$ with reaction constant $k_W$, we may assume that $k_W$ depends on type $I$ of infected individuals -- or additionally on type $D$ of deceased individuals -- and that $k_W$ is a monotonously decreasing function of $I$. Possible examples for $k_W$ are given by 

\begin{align*}
    k_W(I)=\begin{cases}\kappa_W, & I< I_0, \\
    \kappa_R, & I\geq I_0,
    \end{cases}
\end{align*} 
for constants $\kappa_W>\kappa_R>0$ and some $I_0>0$. Alternatively, one may consider 
\begin{equation}
    k_W(I)=\frac{\kappa}{I+1}
    \label{eq:k_W}
\end{equation}
for some positive constant $\kappa$. One can verify that for a rule-based  framework with  type-dependent reaction rates, the standard Gillespie algorithm \ref{alg:sec:model:subsec:gillespie} can be applied.
\end{remark}

\subsection{Gillespie algorithm for discrete subtypes}
\label{sec:model:subsec:gillespie:subsub:subtypes}

In certain modelling scenarios, it may be beneficial to differentiate types further into subtypes. For instance, in the context of age classes, one may subdivide a given population not only in terms of their disease types, but also consider an age structure of individuals in addition. For describing these discrete subtypes, we follow the notation in \cite{rbse}:

\begin{eqnarray}
\label{eq:sec:model:subsec:gillespie:subtypes}
\sum_{i=1}^{s} \gamma_{ij} \sum_{\alpha=1}^{l_i}  \mu_{ij}^\alpha T^\alpha_i  \xrightarrow{k_j} \sum_{i=1}^{s} \delta_{ij} \sum_{\alpha=1}^{l_i}  \nu_{ij}^\alpha T^\alpha_i.
\end{eqnarray}
This is a direct generalisation of the reaction-scheme \eqref{eq:sec:model:subsec:reaction:reactionscheme}. We denote by $l_i\in \N$  the number of subtypes of type $i$ and write $T_i^{\alpha_i}$ for $1\leq \alpha_i\leq l_i$ for the subtypes of type $i$. This notation is generalised to more than one subtype by multi-indices $\xi \in \Xi = \mathbb{N}^{b_1 \times \ldots \times b_u}$, where $u \in \mathbb{N}$ is the total number of different subtypes for all types.

The use of subtypes requires clarification of the Gillespie Algorithm~\ref{alg:sec:model:subsec:gillespie}. This motivates to keep the original algorithm and add a tensor like structure for every subtype dependent quantity, like propensities or parameters, resulting in Algorithm~\ref{alg:sec:model:subsec:gillespie:subsub:subtypes}.

\begin{algorithm}[htb] 
\caption{Gillespie algorithm for discrete subtypes}
\label{alg:sec:model:subsec:gillespie:subsub:subtypes}
\begin{algorithmic}[1] \setcounter{ALG@line}{-1}
\State Initialize the starting time $t$, the system configuration $n=(n^\Xi_1,\ldots,n^\Xi_s)$ and the rate constants $k^\Xi_j$ for $j=1,\ldots,r$. 
\State Generate two random numbers $r_1, r_2$ uniformly distributed in $[0,1]$.\label{alg:line:GillespieSUBFirst}
\State Compute propensities $\kappa^\xi_l (n)$, $\xi \in \Xi$, set $\Phi:= \sum_{l=1}^r \sum_{\xi \in \Xi} \kappa^\xi_l (n)$ and compute $\tau:= -\frac{1}{\Phi} \ln r_1$.
\State Set $t::=t+\tau$ as the time of the next rule execution.
\State Determine the rule $R^{\Xi_z}_i$ which is executed at time $t$ by finding first $i$ such that
\begin{align*}
     \frac{1}{\Phi} \sum_{l=1}^{i-1} \sum_{\xi \in \Xi} \kappa^{\xi}_l (n)<   &r_2 \leq \frac{1}{\Phi} \sum_{l=1}^{i} \sum_{\xi \in \Xi} \kappa^{\xi}_l (n),\\
     \intertext{and then $z$ such that}
     \frac{1}{\Phi} \sum_{l=1}^{i-1} \sum_{q=1}^{z-1} \kappa^{\Xi_q}_l (n)<   &r_2 \leq \frac{1}{\Phi} \sum_{l=1}^{i-1} \sum_{q=1}^{z} \kappa^{\Xi_q}_l (n).
\end{align*}
\State Execute rule $R^{\Xi_z}_i$ and update the new system configuration $n$.
\State \textbf{If} $t<T$ go to step~\ref{alg:line:GillespieSUBFirst}, \textbf{otherwise} stop.
\end{algorithmic}
\end{algorithm}

Only considering discrete subtypes we assume that the set $\Xi$ is ordered, e.g.\ lexicographically, and by $\Xi_q$ we mean the $q^\text{th}$ multi-index.
The first uniformly distributed random number determines the next rule execution time, by inverse method of sampling from an exponential distribution with intensity

\begin{equation}\label{eq:sec:model:gillespie:subsub:subtypePhi}
    \Phi= \sum_{l=1}^r \sum_{\xi \in \Xi} \kappa^\xi_l (n).
\end{equation}

A simple approach to solve the type of rule that is executed, is to flatten the normalised tensor of all propensities into an one dimensional vector while keeping a consistent ordering, and search for the bin, where the second random number $r_2$ lies in.
For large numbers of subtypes this method is prone to numerical errors and can be replaced by tree structures and iterative search in the hierarchical bins.
The updates of the jump processes are executed by adding the row of the transposed of the expanded stoichiometric matrix $\alpha^\Xi$ corresponding to the determined rule.

\subsection{Gillespie algorithm with time-discrete updates}
\label{sec:model:subsec:gillespie:subsub:discreteupdates}

To prevent the spread of the Corona virus, non-pharmaceutical interventions like lockdowns or wearing masks are often imposed for a given discrete time. These measures require updating the parameters of the simulations.
Also monitoring the COVID pandemic, by, e.g., recording the daily number of new infections. Hence, it is desirable to consider a feedback based on these updates, see Remark~\ref{rem:model:subsec:gillespie:subsub:typedependentrates}. The aim of this section is to adapt the Gillespie algorithm~\ref{alg:sec:model:subsec:gillespie} to incorporate these updates.
 
In every time step of Algorithm~\ref{alg:sec:model:subsec:gillespie}  the waiting time $\tau$ for the execution of the next rule   is computed  by 
 \begin{align*}
    \tau \sum_{l=1}^r \kappa_l(n) = -\ln r_1.
 \end{align*}
Since the propensities are time-dependent now, i.e. $\kappa_l=\kappa_l(n(t),t)$, the origin of the time-scale has to be adapted. Let $t_0$ denote the initial time. For a given time $t\geq t_0$, we replace $\tau \kappa_l(n)$ by $\int_{t}^{t+\tau} \kappa_l(n(t'), t')  dt'$
and we solve
\begin{equation}\label{eq:waitingtimeadapt}
    \sum_{l=1}^r \int_t^{t+\tau} \kappa_l(n(t'), t')  dt'= - \ln r_1
\end{equation}
for $\tau$.
For general time-dependent propensity functions, \eqref{eq:waitingtimeadapt} is difficult to solve for $\tau$. To overcome this, we assume that the state $n$ and propensities $\kappa_l$ are piece-wise constant functions which are constant on time intervals $[m\Delta t,(m+1)\Delta t)$ for any $m\in \N$.
For $t\in [m\Delta t,(m+1)\Delta t)$ we have
\begin{align*}
    n(t)=n(m\Delta t)
\end{align*}
and
we define
\begin{equation*}
   \kappa_l^m=\kappa_l(n(m\Delta t),m\Delta t) =\kappa_l(n(t),t). 
\end{equation*}
This shows that $\kappa_l$ is a right-continuous step function as a function of $t$.

Suppose now that $\tau>0$ is known and that the above assumptions on $n$ and $\kappa_l$ hold. 
Then there exist $m_0, m_\tau \in \N$ with $m_0 \leq m_\tau$ such that $t \in [m_0\Delta t, (m_0 +1)\Delta t)$ and $t +\tau \in [m_\tau \Delta t, (m_\tau +1)\Delta t)$.
We introduce the function $\Phi(t)=\sum_{l=1}^r  \kappa_l(n(t), t)$ and   write \eqref{eq:waitingtimeadapt} as $\int_t^{t+\tau} \Phi(t') dt'= - \ln r_1$.
For $m_0=m_\tau$, the left-hand side of \eqref{eq:waitingtimeadapt} satisfies

\begin{equation*}
     \sum_{l=1}^r \int_t^{t+\tau} \kappa_l(n(t'), t')  dt = \sum_{l=1}^r \tau \kappa_l^{m_0} =\tau \Phi(m_0\Delta t),
\end{equation*}
while for $m_0<m_\tau$, we have
\begin{align*}
     &\sum_{l=1}^r \int_t^{t+\tau} \kappa_l(n(t'), t')  dt\\ &= \sum_{l=1}^r \left(  [(m_0+1)\Delta t - t] \kappa_l^{m_0} + \sum_{m=m_0+1}^{m_\tau-1}  \Delta t \kappa_l^{m} +  [t + \tau - m_\tau \Delta t] \kappa_l^{m_\tau}\right)\\
     &= ((m_0+1)\Delta t - t) \Phi(m_0\Delta t) + \sum_{m=m_0+1}^{m_\tau-1}  \Delta t \Phi(m\Delta t) +  (t + \tau - m_\tau \Delta t) \Phi(m_\tau\Delta t),
\end{align*}
where we used that $\kappa_l$ is a right-continuous step function with values $\kappa_l^m$. 
This allows us to rewrite \eqref{eq:waitingtimeadapt} as

\begin{align}
\label{eq:propensityeq}
   -\ln r_1=\begin{cases} \tau \Phi(m_0\Delta t), & m_\tau=m_0,\\
    ((m_0+1)\Delta t - t) \Phi(m_0\Delta t)\\
    \quad + (t + \tau - m_\tau \Delta t) \Phi(m_\tau\Delta t), & m_\tau = m_0 + 1,\\
   ((m_0+1)\Delta t - t) \Phi(m_0\Delta t)\\
   \quad + \sum_{m=m_0+1}^{m_\tau-1}  \Delta t \Phi(m\Delta t) \\
   \quad +  (t + \tau - m_\tau \Delta t) \Phi(m_\tau\Delta t), & m_\tau > m_0+1.
   \end{cases}
\end{align}

Next, we describe an iterative procedure to determine $\tau$ from \eqref{eq:waitingtimeadapt} under the assumption that $\kappa_l$ is piece-wise constant.
Let $t\geq t_0$ and $m_0\in \mathbb{N}$ be given such that $t\in [m_0\Delta t,(m_0+1)\Delta t)$. We predict $\tau$ as the solution of 
\begin{equation}\label{eq:waitingtimeadapt2}
   \tau  \Phi(m_0\Delta t)= - \ln r_1,
\end{equation}
provided $0<\tau\leq (m_0+1)\Delta t-t$. 
 If $\tau>(m_0+1)\Delta t-t$ holds for $\tau$ satisfying \eqref{eq:waitingtimeadapt2}, we reject the previously computed $\tau$, set  $m_\tau =m_0+1$ and obtain
 
\begin{equation*}
    ((m_0+1)\Delta t-t)  \Phi(m_0\Delta  t)+(t+\tau -(m_0+1)\Delta t)  \Phi((m_0+1)\Delta t)= - \ln r_1,
\end{equation*}
or equivalently, 

\begin{equation*}
    \tau  = - \frac{\ln r_1}{\Phi((m_0+1)\Delta t)}-\frac{((m_0+1)\Delta t-t)  \Phi(m_0\Delta  t)}{\Phi((m_0+1)\Delta t)}+(m_0+1)\Delta t-t,
\end{equation*}
provided $\tau \in((m_0+1)\Delta t-t,(m_0+2)\Delta t-t]$. If $\tau>(m_0+2)\Delta t-t$, 
the required condition can be derived iteratively. The above procedure illustrates the idea for our algorithm. First, we solve \eqref{eq:propensityeq} for $m_\tau=m_0$ which is equivalent to \eqref{eq:waitingtimeadapt2}. While $\tau >(m_\tau+1)\Delta t-t$, we increase $m_\tau$ by 1 and solve \eqref{eq:propensityeq} for $\tau$. 

Suppose now that $\tau$ has been computed by the procedure above. We choose $m_\tau\in \mathbb{N}$ such that $t+\tau\in [m_\tau \Delta t,(m_\tau+1)\Delta t)$.
To determine which rule is applied  we consider the probability that reaction $R_j$ ever occurs  which is denoted by $\rho_j$ as in \eqref{eq:rho_i_standardgil}. 
Below in Lemma~\ref{lem:estimatejustification}, we show that 

\begin{align*}
\rho_j(n(t+\tau),t+\tau)=\frac{\kappa_j(n(t+\tau),t+\tau)}{\sum_{l=1}^r\kappa_l(n(t+\tau),t+\tau)}=\frac{\kappa_j(n(t+\tau),t+\tau)}{\Phi(t+\tau)}    
\end{align*}
for time-dependent states $n$ and $\kappa_l$ for $l=1,\ldots,r$. 
The previous relation states that once the  execution time $\tau$ of the next reaction is determined, the probability of a reaction to be executed at time $t+\tau$ 
depends on the state and the propensities at  time $t+\tau$.
Using \eqref{eq:estimate} from Lemma~\ref{lem:estimatejustification}, we can rewrite condition \eqref{eq:conditiongillespie} 
in the original Gillespie algorithm~\ref{alg:sec:model:subsec:gillespie}:
Given a uniformly distributed random number $r_2\in [0,1]$, we chose $i\in\{1,\ldots,r\}$ such that

\begin{equation*}
    \frac{1}{\Phi(t + \tau)} \sum_{l=1}^{i-1} \kappa_l(n(t + \tau),t + \tau) < r_2 \leq \frac{1}{\Phi(t + \tau)} \sum_{l=1}^{i} \kappa_l(n(t + \tau),t + \tau)
\end{equation*}
is satisfied. To evaluate this condition, note that $ \kappa_l(n(t + \tau), t + \tau)=\kappa_l(n(m_\tau \Delta t), m_\tau \Delta t) $ 
for $t+\tau \in [m_\tau \Delta t,(m_\tau +1)\Delta t)$ due to the assumption on the piece-wise constant propensities and states on intervals $[m\Delta t,(m+1)\Delta t)$ for $m\in\N$. We summarise the computation of $\tau$ and the determination of reaction $R_j$ in Algorithm~\ref{alg:sec:model:subsec:gillespie:subsub:discreteupdates}.
\clearpage
\begin{algorithm}[htb]
\caption{Gillespie algorithm with time-discrete updates}
\label{alg:sec:model:subsec:gillespie:subsub:discreteupdates}
\begin{algorithmic}[1] \setcounter{ALG@line}{-1} 
\State Initialize the starting time $t$, the system configuration $n=(n_1,\ldots,n_s)$ and the rate functions $k_l(t)$ for $l=1,\ldots,r$ as well as the discrete time-step $\Delta t$ and a counter $m_0\in \N$ such that $t\in[m_0\Delta t,(m_0+1)\Delta t)$.
\State Generate two random numbers $r_1, r_2$ uniformly distributed in $[0,1]$. \label{alg:line:GillespieDUrandom}
\algrenewcommand\alglinenumber[1]{\footnotesize #1a:}
\State Set $m_\tau=m_0$, $\Phi(t)= \sum_{l=1}^r \kappa_l (n(t),t)$ and compute  ${\tau =- \frac{\ln r_1}{\Phi(t)}}$. \addtocounter{ALG@line}{-1}
\algrenewcommand\alglinenumber[1]{\footnotesize #1b:}
\While{$t + \tau > (m_\tau+1) \Delta t$}: 
\StateXfirst Set $m_\tau ::= m_\tau+1$ and
    \begin{align*}
    \tau  &= - \frac{\ln r_1}{\Phi(m_r\Delta t)}-\frac{((m_0+1)\Delta t-t)  \Phi(m_0\Delta  t)}{\Phi(m_r\Delta t)}\\
        & \qquad-\sum_{m=m_0+1}^{m_\tau-1}  \Delta t \frac{\Phi(m\Delta t)}{\Phi(m_r\Delta t)}+m_r\Delta t-t.
    \end{align*}
    \EndWhile
\algrenewcommand\alglinenumber[1]{\footnotesize #1:}
\State Set $ t::= t + \tau$ as the time of the next rule execution. \label{alg:line:GillespieDUupdateTime}
\State Determine the rule $R_i$ which is executed at time $t$ by finding $i$ such that
\begin{equation}
     \frac{1}{\Phi(t)} \sum_{l=1}^{i-1} \kappa_l(n(t), t) < r_2 \leq \frac{1}{\Phi(t)} \sum_{l=1}^{i} \kappa_l(n(t), t).
\end{equation}

\State Execute rule $R_i$ and update the new system configuration $n$.
\State \textbf{If} $t<T$ go to step~\ref{alg:line:GillespieDUrandom}, \textbf{otherwise} stop.
\end{algorithmic}
\end{algorithm}
 
Proceeding as in \cite[Appendix A]{Anderson2007}, the following result holds in our setting:

\begin{lemma}
\label{lem:estimatejustification}
Given $t\geq t_0$ and $\tau>0$, the probability  that reaction $R_j$ is executed after waiting time $\tau$ is
\begin{align}
    \rho_j(n(t+\tau),t+\tau)=\frac{\kappa_j(n(t+\tau),t+\tau)}{\sum_{l=1}^r\kappa_l(n(t+\tau),t+\tau)}.
    \label{eq:estimate}
\end{align}
\end{lemma}

\begin{proof}
We follow the approach in \cite[Appendix A]{Anderson2007}. For this, we introduce $r$ 
 independent exponentially distributed random variables $\tau_l$ for $l=1,\ldots,r$ where $\tau_l$ denotes the waiting time until reaction $R_l$ occurs. We assume that reaction $R_j$ is executed if it has the shortest waiting time among all reaction waiting times $\tau_l,~l=1,\ldots,r$, and in this case it is executed instantaneously. We have 

 \begin{align}\label{eq:proofrem}
 \rho_j(n(t+\tau),t+\tau)=P(\tau_j < \tau_{l}~ \forall l \neq j \vert \min_k\{\tau_k\}=\tau),
\end{align} 
i.e., $\rho_j$ is 
the conditional probability that  reaction $R_j$ is the first reaction to be executed, given that $\tau = \min_l \tau_l$ is the waiting time for the first reaction to be executed. 
The random variables $\tau_l$ for $l\in\{1,\ldots,r\}$ satisfy

\begin{equation*}
P(\tau_l>\tau)=\exp\left(-\int_t^{t+\tau}\kappa_l(n(t'),t')\di t'\right).
\end{equation*}
 Since $\min_l \tau_l>\tau$ is equivalent to $\tau_l>\tau$ for all $l\in\{1,\ldots,r\}$, the independence of $\tau_l$ implies
 
\begin{equation}
\label{eq12}
P(\min_l \tau_l>\tau)=\prod_{l=1}^r P(\tau_l>\tau)=\exp\left(-\int_t^{t+\tau}\sum_{l=1}^r\kappa_l(n(t'),t')\di t'\right).
\end{equation}
The right-hand side of \eqref{eq:proofrem} can be written as an appropriate limiting probability given by

\begin{align*}
P(\tau_j < \tau_{l}\;\forall l\neq j\;\vert\min_k\{\tau_k\}=\tau)&=
\lim_{\delta \rightarrow 0}P(\tau_j < \tau_{l}\;\forall l \neq j\;\vert \min_k\{\tau_k\}\in[\tau,\tau +\delta))\\
&=\lim_{\delta \rightarrow 0}\frac{P(\tau_j<\tau_{l}\;\forall l \neq j,\;\min_k\{\tau_k\} \in [\tau,\tau+\delta) )}{P(\min_k\{\tau_k\}\in[\tau,\tau+\delta))},
\end{align*}
where the second equality follows from the definition of the conditional expectation.
Since the conditions $\tau_j<\tau_{l}\;\forall l\neq j$ and $\min_k\{\tau_k\}\in[\tau,\tau+\delta)$ are equivalent to $\tau_j\in[\tau,\tau+\delta)$ and $\tau_l>\tau+\delta\,\,\forall l\neq j$, we obtain

\begin{align*}
P(\tau_j < \tau_{l}\;\forall l\neq j\;\vert\min_k\{\tau_k\}=\tau)&=\lim_{\delta\to 0}
\frac{P(\tau_j\in[\tau,\tau+\delta), \tau_l>\tau+\delta\,\,\forall l\neq j)}{P(\min_k\{\tau_k\}\in[\tau,\tau+\delta))},
\end{align*}
which yields
\[
P(\tau_j < \tau_{l}\;\forall l\neq j\;\vert\min_k\{\tau_k\}=\tau)=\lim_{\delta\to 0}
\frac{P(\tau_j\in[\tau,\tau+\delta))\prod_{\substack{l = 1\\ l\neq k}}^r P(\tau_l>\tau+\delta)}{P(\min_k\{\tau_k\}\in[\tau,\tau+\delta))}\]
by the independence of $\tau_l$ for $l=1,\ldots,r$.
Since

\begin{align*}
&P(\min_k \tau_k\in[\tau,\tau+\delta))\\&=P(\min_k \tau_k\geq \tau)-P(\min_k \tau_k\geq \tau +\delta)\\&=
\exp\left(-\int_t^{t+\tau}\sum_{l = 1}^r\kappa_l(n(t'),t')\di t'\right)-
\exp\left(-\int_t^{t+\tau + \delta}\sum_{l=1}^r\kappa_l(n(t'),t')\di t'\right).
\end{align*}
To simplify the notation let us set

\begin{equation}
    \label{eq:int_kappa}
    f(\tau;\kappa_j)=\exp\left(-\int_t^{t+\tau}\kappa_j(n(t'),t')\di t'\right).
\end{equation}
By \eqref{eq12}, we obtain

\begin{align*}
P(\tau_j < \tau_{l}\;\forall l\neq j\;\vert\min_k\{\tau_k\}=\tau)=&\lim_{\delta\rightarrow 0}%
\frac{\left(f(\tau;\kappa_j)-f(\tau+\delta;\kappa_j)\right)\prod_{\substack{l = 1\\ l\neq j}}^rf(\tau+\delta;\kappa_l)}%
{\prod_{l=1}^rf(\tau;\kappa_l)-\prod_{l=1}^rf(\tau+\delta;\kappa_l)}\\[2mm]
=&\frac{\kappa_j(n(t+\tau),t+\tau)}{\sum_{l=1}^r\kappa_l(n(t+\tau),t+\tau)}
\end{align*}
by L'Hôpital's rule. This concludes the proof of \eqref{eq:estimate}.
\end{proof}

The above description of the iterative procedure in Algorithm~\ref{alg:sec:model:subsec:gillespie:subsub:discreteupdates} shows that evaluations of $\Phi$ can be reduced by propensity accumulation which makes the implementation more efficient, especially if the time scales require several discrete updates within the next waiting time. First, we solve \eqref{eq:waitingtimeadapt2} and set $\hat{\Phi}$ as $-\ln r_1 - ((m_0+1)\Delta t-t)\Phi(m_0\Delta t)$ and $m_\tau=m_0$. While $\tau >(m_\tau+1)\Delta t-t$, we increase $m_\tau$ by 1, 
compute $\tau$ by
\begin{align*}
    \tau =\frac{\hat{\Phi}}{\Phi(m_\tau \Delta t)}+m_\tau \Delta t-t
\end{align*}
and set $\hat \Phi$ as $\hat \Phi-\Delta t \Phi(m_\tau \Delta t)$.

Algorithm~\ref{alg:sec:model:subsec:gillespie:subsub:discreteupdates2v2} for time-discrete updates and propensity accumulation overcomes the inefficient computations in Algorithm \ref{alg:sec:model:subsec:gillespie:subsub:discreteupdates} and produces a time series for reaction events only. As we want to keep track of updates not only by reactions but also by global updates, we treat discrete time steps equitably. Suppose that the last reaction occurred at time $t_0$. We define 

\begin{equation}
    \tilde{\tau} := 
        \begin{cases}
        \tau, & m_\tau = m_0,\\
        t_0 + \tau - m_\tau \Delta t, & m_\tau > m_0,
        \end{cases}
\end{equation} 
as the update to the last time step (either from reaction or global). For  $0 \leq m \leq m_\tau - m_0 -1$, we define a accumulation variable 

\begin{equation}\label{eq:DefPhihat}
    \hat\Phi_m :=
    \begin{cases}
        ((m_0 + 1) \Delta t - t_0) \Phi(m_0 \Delta t), & m = 0,\\
        \Delta t \Phi((m_0 + m)\Delta t), & m > 0,
    \end{cases}
\end{equation}
which stores the values of the integral of the piece-wise constant propensities.
For convenience we introduce $v = - \ln{r_1}$, which can be interpreted as a virtual normalised time. This allows us to write \eqref{eq:propensityeq} as

\begin{equation}\label{eq:propensityeqPhihat}
 v =  \sum_{m = 0}^{m_\tau - m_0 - 1} \hat\Phi_m + \tilde\tau \Phi(m_\tau \Delta t)
\end{equation}
For $m_\tau=m_0$, \eqref{eq:propensityeqPhihat} reduces to 

\begin{equation}\label{eq:propensityeqPhihatm0}
    v =  \tau \Phi(m_0 \Delta t).
\end{equation}
In  Algorithm~\ref{alg:sec:model:subsec:gillespie:subsub:discreteupdates2v2} we initially 
solve \eqref{eq:propensityeqPhihatm0}  for $\tau$ and set $\tilde \tau =\tau$. We  initialize the accumulation variable $\hat\Phi$ for $\sum_{m=0}^{m_\tau-m_0-1}\hat\Phi_m$ by setting $\hat{\Phi} = 0$ and $m_\tau = m_0$. If 
$\tilde\tau$ satisfies $t_0 + \tilde\tau \geq (m_\tau+1)\Delta t$, 
we increase $m_\tau$ by 1, replace $\hat \Phi$ by $\hat \Phi+\hat \Phi_{m_\tau-m_0-1}$, set time $t$ as $m_\tau \Delta t$ and compute the new prediction $\tilde\tau$ from \eqref{eq:propensityeqPhihat}, given by

\begin{align*}
    \tilde\tau =\frac{v - \hat\Phi}{\Phi(m_\tau \Delta t)}.
\end{align*}
Time $t$ is not a result of a reaction but of a discrete update, where changes in states could also happen. We repeat this procedure and if the predicted waiting time increment $\tilde\tau$ lies in the actual interval, we proceed with step~\ref{alg:line:GillespieDUupdateTime} of Algorithm~\ref{alg:sec:model:subsec:gillespie:subsub:discreteupdates}.

\begin{algorithm}[ht]
\caption{Gillespie algorithm with time-discrete updates and propensity accumulation}
\label{alg:sec:model:subsec:gillespie:subsub:discreteupdates2v2}
\begin{algorithmic}[1] \setcounter{ALG@line}{-1} 
\State Initialize the starting time $t$, the system configuration $n=(n_1,\ldots,n_s)$ and the rate functions $k_l(t)$ for $l=1,\ldots,r$ as well as the discrete time-step $\Delta t$ and $m_0\in \N$ such that $t\in[m_0\Delta t,(m_0+1)\Delta t)$.
\State Generate two random numbers $r_1, r_2$ uniformly distributed in $[0,1]$ and compute $v= - \ln r_1$.
\label{alg:line:GillespieDUBKrandom}
\algrenewcommand\alglinenumber[1]{\footnotesize #1a:}
\State Compute sum of propensities $\Phi(t)= \sum_{l=1}^r \kappa_l (n(t),t)$ and prediction ${\tilde\tau := \frac{v}{\Phi(t)}}$.
\Statex Set $m = m_0$, $\hat{\Phi} = 0$.  \addtocounter{ALG@line}{-1}
\algrenewcommand\alglinenumber[1]{\footnotesize #1b:}
\While{$t + \tilde\tau \geq (m+1) \Delta t$}: \label{alg:line:GillespieDUBKwhile}
    \StateXmulti{Set $\hat{\Phi} ::= \hat{\Phi} + \Phi(m \Delta t) \cdot \left( (m+1) \Delta t - t\right)$ for accumulation.
    Set $m ::= m+1$ and $t ::= m \Delta t$.
    Update propensities due to global update and compute new prediction
    \begin{equation*}
        \tilde\tau = \frac{v - \hat{\Phi}}{\Phi(m \Delta t)}.
    \end{equation*}}
\EndWhile
\algrenewcommand\alglinenumber[1]{\footnotesize #1:}
\State Set $ t::= t + \tilde\tau$ as the time of the next rule execution.
\State Determine the rule $R_i$ which is executed at time $t$ by finding $i$ such that
\begin{equation}
     \frac{1}{\Phi(t)} \sum_{l=1}^{i-1} \kappa_l(n(t), t) < r_2 \leq \frac{1}{\Phi(t)} \sum_{l=1}^{i} \kappa_l(n(t), t).
\end{equation}

\State Execute rule $R_i$ and update the new system configuration $n$.
\State \textbf{If} $t<T$  go to step~\ref{alg:line:GillespieDUBKrandom}, \textbf{otherwise} stop.
\end{algorithmic}
\end{algorithm}

\begin{remark}
\label{rem:discreteUpdatesBookKeepingTimestep}
For the time-step in the update of propensity accumulation $\hat{\Phi}$ in step~\ref{alg:line:GillespieDUBKwhile}b of Algorithm~\ref{alg:sec:model:subsec:gillespie:subsub:discreteupdates2v2}, $\left( (m_\tau +1) \Delta t - t\right) = \Delta t$  holds after the first while loop, so this part covers both cases of the definition of $\hat\Phi_m$ in \eqref{eq:DefPhihat}.
\end{remark}

\begin{remark}
\label{rem:discreteUpdatesNonEquidistant}
All considerations above can be adapted straight forward for discrete time steps which are not equidistant - even the acceleration by accumulation.
\end{remark}

The iterative procedure of finding by prediction and rejection is visualised in Figure~\ref{fig:GAdiscUpV2}.

\begin{figure}[ht]
    \centering
    \includegraphics[width=0.9\linewidth]{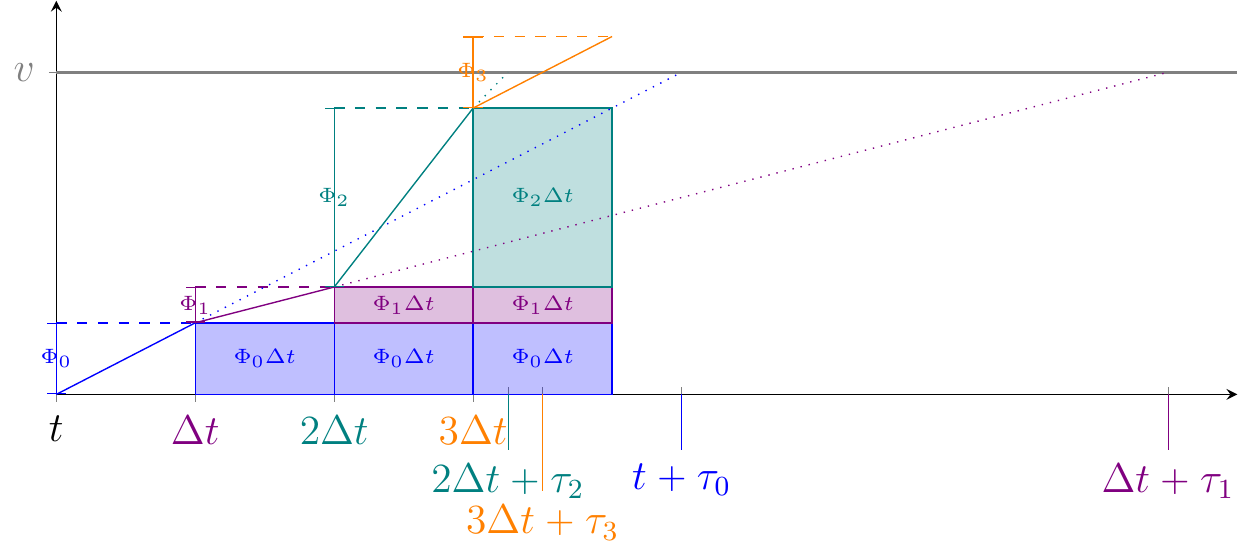}
    \caption{Visualisation of Algorithm~\ref{alg:sec:model:subsec:gillespie:subsub:discreteupdates2v2} for $t=0$ and $\Delta t = 1$; rectangles illustrate growing $\hat\Phi$, dotted lines denote rejected predictions (${\color{blue}\tau_0}, {\color{violet}\tau_1}, {\color{teal}\tau_2}$). The next correct time is $t + {\color{orange}3 \Delta t + \tau_3}$.}
    \label{fig:GAdiscUpV2}
\end{figure}

\subsection{Gillespie algorithm with deterministic delays}
\label{sec:model:subsec:gillespie:subsub:delay}

The idea of predicting tentative execution times as in Section~\ref{sec:model:subsec:gillespie:subsub:discreteupdates} can be applied to deterministic delays for rules, e.g., there is a deterministic waiting time until a specific change in types such as release from isolation occurs or the imposed waiting time for a second vaccination.
We extend the notation for rules in \eqref{eq:sec:model:subsec:reaction:reactionscheme} by delays

\begin{eqnarray}
\label{eq:sec:model:subsec:gillespie:subsub:delayedRule}
\sum_{i=1}^{s} \alpha_{ij}  T_i  \xrightarrow{k_j} \sum_{i=1}^{s} \beta_{ij}  T_i \mathop{\dashrightarrow}^{}_{\omega_j} \sum_{i=1}^{s} \gamma_{ij}  T_i,
\end{eqnarray}
where $\omega_j \in \mathbb{R}_0^+$ denote waiting times.
In Algorithm~\ref{alg:sec:model:subsec:gillespie:subsub:deterministicDelays} we use queues $Q^{(q)}$ for appending future delayed execution times of each rule with delays.
The queues $Q^{(q)}$ might be non-empty at the beginning.

\begin{algorithm}[ht]
\caption{Gillespie algorithm with deterministic delays}
\label{alg:sec:model:subsec:gillespie:subsub:deterministicDelays}
\begin{algorithmic}[1] \setcounter{ALG@line}{-1} 
\State Initialize the starting time $t$, the system configuration $n=(n_1,\ldots,n_s)$ and the rate functions $k_l(t)$ for $l=1,\ldots,r$ as well as the delay queues $Q^{(q)}$.
\State Generate two random numbers $r_1, r_2$ uniformly distributed in $[0,1]$ and compute $v= - \ln r_1$. \label{alg:line:GillespieDUDDrandom} 
\algrenewcommand\alglinenumber[1]{\footnotesize #1a:}
\State Compute sum of propensities $\Phi(t)= \sum_{l=1}^r \kappa_l (n(t),t)$ and prediction ${\tilde\tau := \frac{v}{\Phi(t)}}$.
\Statex \textbf{If} all $Q^{(q)}$ are non-empty \textbf{then} set $d ::= \min_q Q^{(q)}_1$ and remove entry in this queue, \textbf{else} set $\delta d ::= \infty$.
\Statex Set $\hat{\Phi} = 0$. \addtocounter{ALG@line}{-1}
\algrenewcommand\alglinenumber[1]{\footnotesize #1b:}
\While{$t + \tilde\tau \geq d$}:
    \StateXfirst Set $\hat{\Phi} ::= \hat{\Phi} + \Phi(t) \cdot \left(d - t\right)$ for accumulation. Set $t ::= d$.
    \StateX Update propensities due to delayed update and compute new prediction
    \begin{equation*}
        \tilde\tau = \frac{v - \hat{\Phi}}{\Phi(t)}.
    \end{equation*}
    \StateXmulti{\textbf{If} all $Q^{(q)}$ are non-empty\textbf{then} set $d ::= \min_q Q^{(q)}_1$ and remove entry in this queue, \textbf{else} set $\delta d ::= \infty$}.
\EndWhile
\algrenewcommand\alglinenumber[1]{\footnotesize #1:}
\State Set $ t::= t + \tilde\tau$ as the time of the next rule execution.
\State Determine the rule $R_i$ which is executed at time $t$ by finding $i$ such that
\begin{equation}
     \frac{1}{\Phi(t)} \sum_{l=1}^{i-1} \kappa_l(n(t), t) < r_2 \leq \frac{1}{\Phi(t)} \sum_{l=1}^{i} \kappa_l(n(t), t).
\end{equation}

\State Execute rule $R_i$, possibly with appending delayed execution times to the queues, and update the new system configuration $n$.
\State \textbf{If} $t<T$ go to step~\ref{alg:line:GillespieDUDDrandom}, \textbf{otherwise} stop.
\end{algorithmic}
\end{algorithm}

\newpage
\section{Numerical simulations}
\label{sec:simulations}

In this section, we consider different examples of rule-based epidemiological models and simulate them based on the stochastic algorithms in Section~\ref{sec:gillespie}.

\subsection{SIRD}
\label{sec:simulations:SIRD}


As we focus on simulating stochastic models of epidemics, we start with the simple SIRD model with types in Table~\ref{tab:sec:standard:subsec:sird:types} and rules in~\eqref{eq:sec:standard:subsec:sir:rules}, where $\iota > 0$, $\rho > 0$ and $\delta > 0$ are the rate parameters for infection, recovery and death.
The propensities are %
$\kappa_1(\nS, \nI, \nR, \nD) = \frac{\iota}{N} \nS \nI$, %
$\kappa_2(\nS, \nI, \nR, \nD) = \rho \nI$, %
$\kappa_3(\nS, \nI, \nR, \nD) = \delta \nI$, %
where the symbols for types represent in this context the number of individuals of resp. type. $N$ denotes the total size of the population.
For the parameters we choose a rather small population of $N = 1000$ to highlight the stochastic effects. Further parameters are conform to~\cite{rbse} and stated in Figure~\ref{fig:plotSIRD(SSA)}.

\noindent
\begin{minipage}[t]{0.45\textwidth}\vspace{0pt}%
\vspace{-7.5pt}
\captionsetup{type=table}%
\caption{Types of the standard SIRD model.}
\label{tab:sec:standard:subsec:sird:types}
\begin{center}
\extrarowheight=0pt
\addtolength{\extrarowheight}{\belowrulesep}
\aboverulesep=0pt
\belowrulesep=0pt
\begin{tabular}{c  l} 
	\toprule
	Type & Interpretation  \\
	\hline
	\cellcolor{tabgreen!80}$S$ & Susceptible  \\
	\cellcolor{taborange!80}$I$ & Infectious  \\
	\cellcolor{tabblue!80}$R$ & Recovered  \\
	\cellcolor{black!80}\color{white}$D$ & Deceased  \\
	\bottomrule
\end{tabular}
\end{center}
\end{minipage}%
\hfill%
\begin{minipage}[t]{0.45\textwidth}\vspace{0pt}%
Rules of the standard SIRD model:

\begin{equation}
\label{eq:sec:standard:subsec:sir:rules}
\begin{alignedat}{3}
	S + I &\xrightarrow{\iota}  && I + I, &&\quad\text{[Infection]}\\
	    I &\xrightarrow{\rho}   && R,     &&\quad\text{[Recovery]}\\
	    I &\xrightarrow{\delta} && D,     &&\quad\text{[Death]}\\
\end{alignedat}
\end{equation}
\end{minipage}
\vspace{10pt}

\begin{figure}[ht]
    \centering
    \begin{subfigure}{0.32\linewidth}
        \centering
        \includegraphics[width=\linewidth]{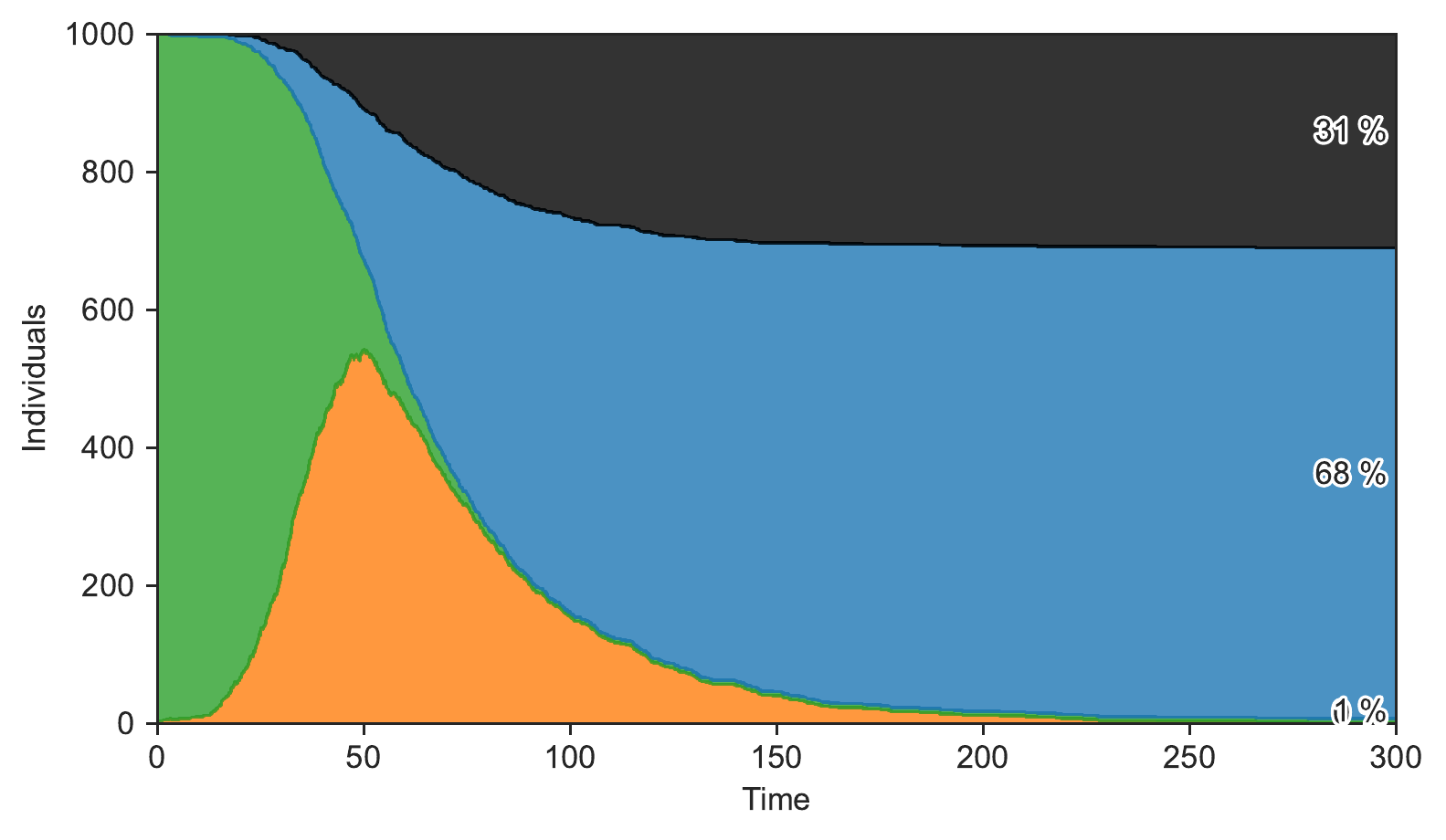}
        \caption{seed $= 0$}
        \label{fig:plotSIRD(SSA)1}
    \end{subfigure}
    \hfill
    \begin{subfigure}{0.32\linewidth}
        \centering
        \includegraphics[width=\linewidth]{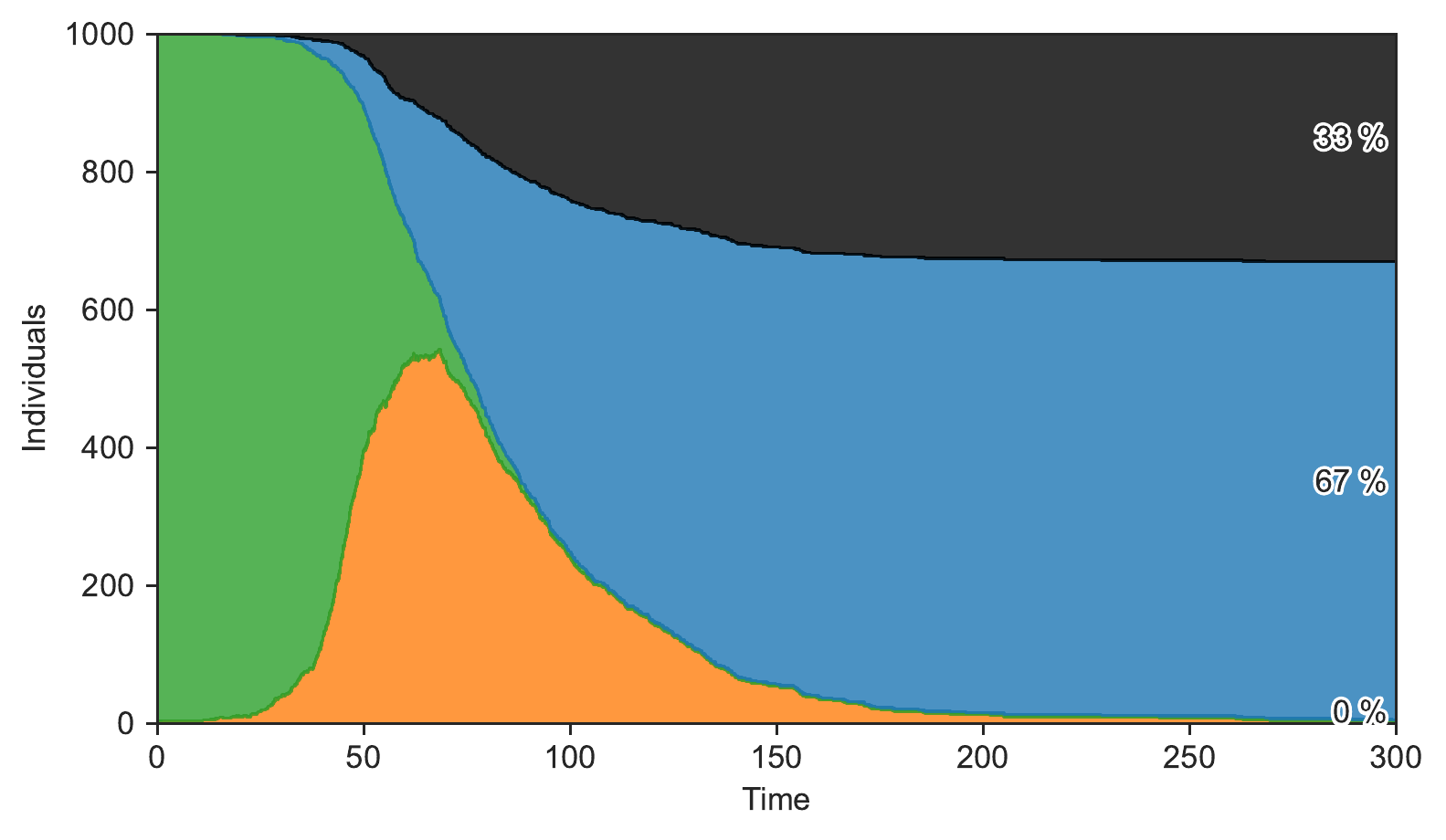}
        \caption{seed $= 1$}
        \label{fig:plotSIRD(SSA)2}
    \end{subfigure}
    \hfill
    \begin{subfigure}{0.32\linewidth}
        \centering
        \includegraphics[width=\linewidth]{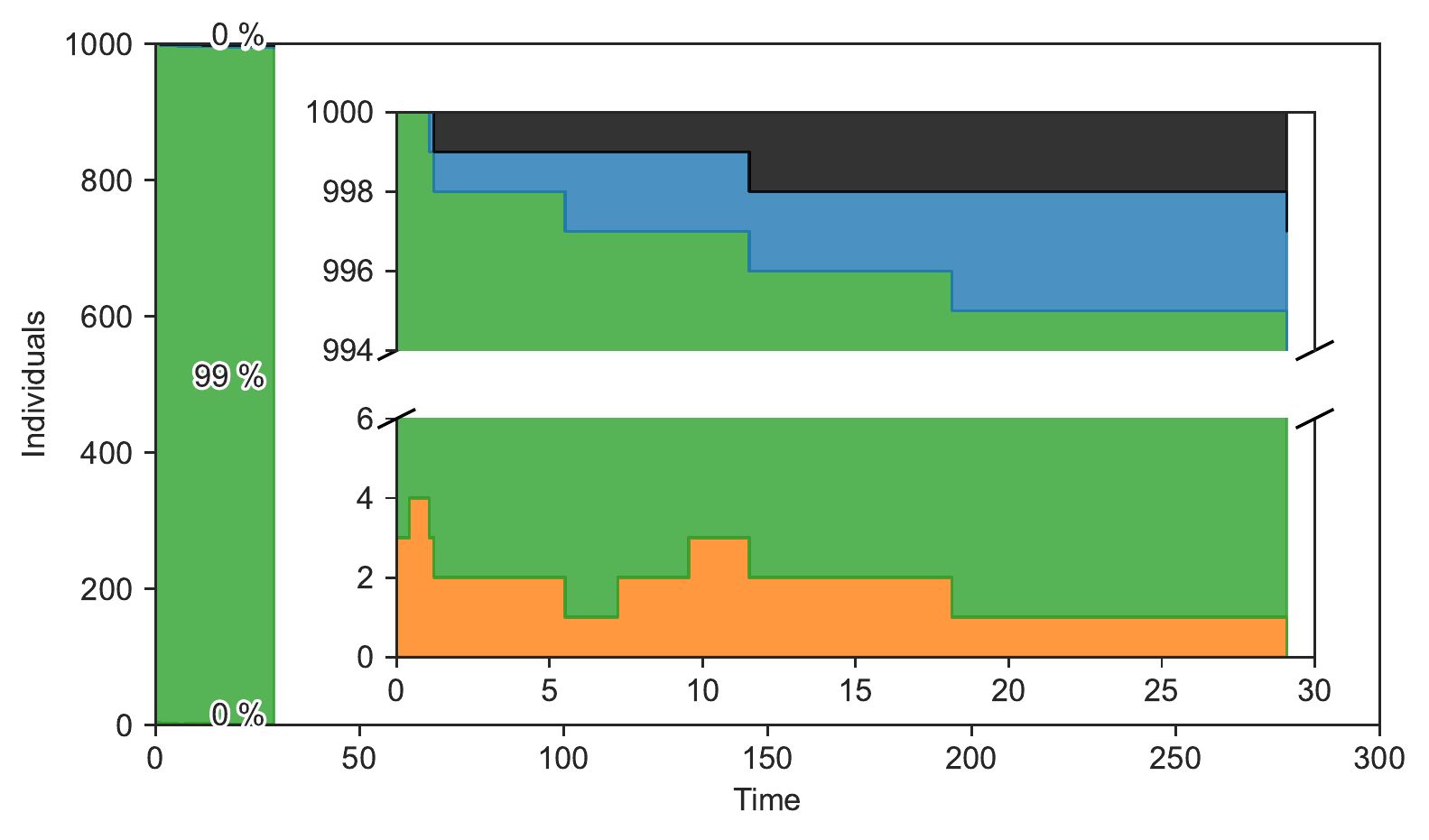}
        \caption{seed $= 11922$}
        \label{fig:plotSIRD(SSA)3}
    \end{subfigure}
    \caption{Three realisations of SSA simulation by Algorithm~\ref{alg:sec:model:subsec:gillespie}  with basic parameters $N = 1000$, $\iota = \iota^0 = 0.2$, $\rho = \rho^0 = 0.02$, $\delta = \delta^0 = 0.01$, $\nI_0 = \nI_0^0 = 3$.}
    \label{fig:plotSIRD(SSA)}
\end{figure}

Stochastic simulations offer a range of realisations, see Figure~\ref{fig:plotSIRD(SSA)}. Some are very near to ODE simulations \subref{fig:plotSIRD(SSA)1}. The randomness dominant in the beginning can also lead to a delay \subref{fig:plotSIRD(SSA)2} or even prevent an outbreak \subref{fig:plotSIRD(SSA)3} (with enlarged inset). For reproducibility we explicitly state the seeds for the random number generator. 

In Figure~\ref{fig:plotSIRD(SSA)_extra} we depict additional plots for the governing quantities of the Gillespie algorithm. The propensities change a lot with growing numbers of infectious individuals, determining the time between events shown a barcodes / rug plot (step~\ref{alg:line:GillespieTwo} and \ref{alg:line:GillespieThree} in Algorithm~\ref{alg:sec:model:subsec:gillespie}). The normalised propensities correspond to the probabilities of which the specific rules are executed (step~\ref{alg:line:GillespieFour} and \ref{alg:line:GillespieFive}) -- thus the type of event represented by the colour and row of the stroke in the barcode.

\begin{figure}[ht]
    \centering
    \includegraphics[width=\linewidth]{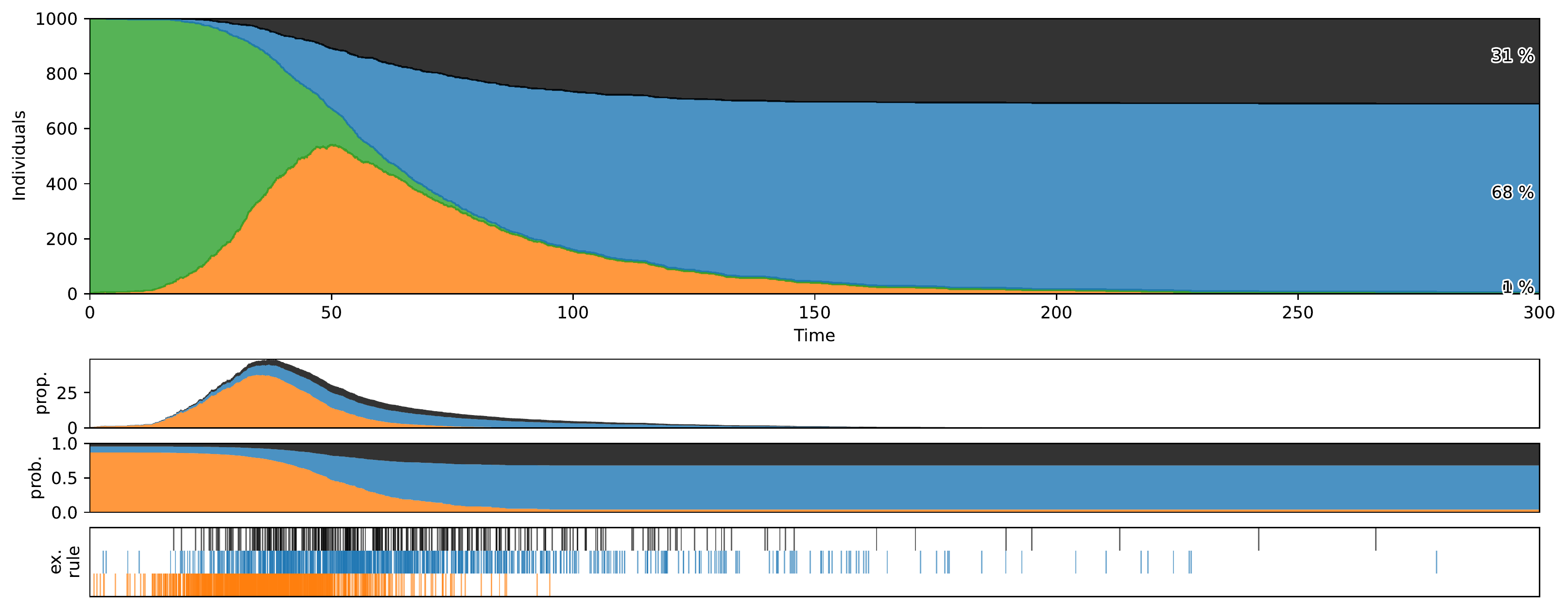}
    \caption{Realisation of SSA simulation by Algorithm~\ref{alg:sec:model:subsec:gillespie} for SIRD model with seed $= 0$. With propensities, probabilities and executed rules as event barcodes.}
    \label{fig:plotSIRD(SSA)_extra}
\end{figure}

By construction the Gillespie algorithm~\ref{alg:sec:model:subsec:gillespie} produces waiting times between consecutive events which are exponentially distributed with intensities depending on the number of individuals. The same holds for the inter event times between events of the same type.\footnote{Sometimes ambiguously stated in other literature as "exponentially distributed with rate parameter"; thus neglecting the dependency on the number of individuals resulting from the propensities.}
The low number of infectious at the beginning and end of the epidemic cause more likely larger event times, see the dark lines in Figure~\ref{fig:plotSIRD(SSA)_interEventTimes}, also observable in Figure~\ref{fig:plotSIRD(SSA)_extra}. To demonstrate that the inter event times are exponentially distributed, we use the closure under scaling and rescale the inter event times by factors from the propensities

\begin{alignat}{4}
    \widetilde{\Delta t^I_e} &= \frac{S I}{N} \Delta t^I_e &=&& \frac{S I}{N} &&(t^I_{e+1} - t^I_{e}),\\
    \widetilde{\Delta t^R_e} &= I \Delta t^R_e &=&& I &&(t^R_{e+1} - t^R_{e}),\\
    \widetilde{\Delta t^D_e} &= I \Delta t^D_e &=&& I &&(t^D_{e+1} - t^D_{e}).
\end{alignat}
Here, the index $e$ corresponds to events where  types are denoted in the superscript -- we omit the first $I_0 = 3$ remove events (recovery or death) to compensate the initial infectious individuals.
Thus the scaled inter event times are exponentially distributed with fixed intensity, i.e.    $\widetilde{\Delta t^I_e} \sim \Exp(\iota)$, $\widetilde{\Delta t^R_e} \sim \Exp(\rho)$ and $\widetilde{\Delta t^D_e} \sim \Exp(\delta)$.
This can also be seen in Figure~\ref{fig:plotSIRD(SSA)_interEventTimesDistributions}.


\begin{figure}[ht]
    \centering
    \begin{subfigure}{0.32\linewidth}
        \centering
        \includegraphics[width=\linewidth]{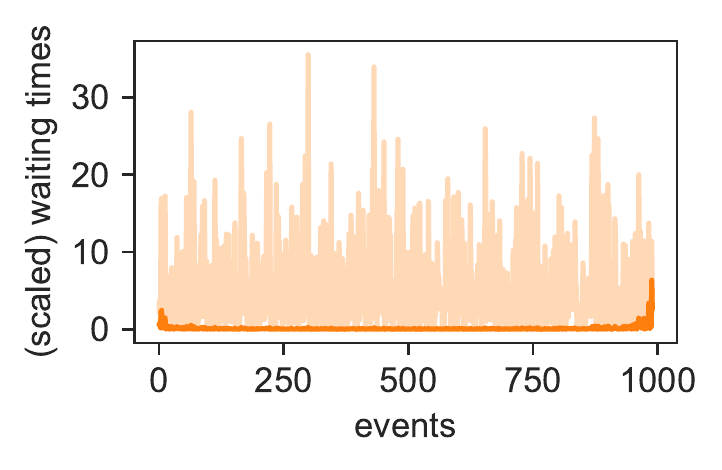}
        \caption{Infection}
        \label{fig:plotSIRD(SSA)_interEventTimesI}
    \end{subfigure}
    \hfill
    \begin{subfigure}{0.32\linewidth}
        \centering
        \includegraphics[width=\linewidth]{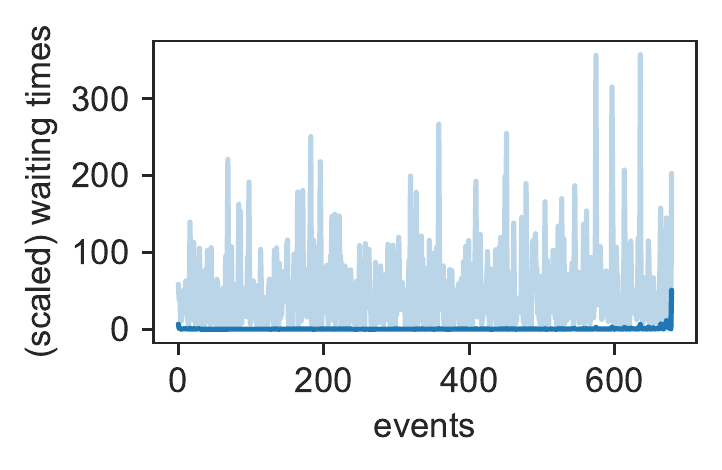}
        \caption{Recovery}
        \label{fig:plotSIRD(SSA)_interEventTimesR}
    \end{subfigure}
    \hfill
    \begin{subfigure}{0.32\linewidth}
        \centering
        \includegraphics[width=\linewidth]{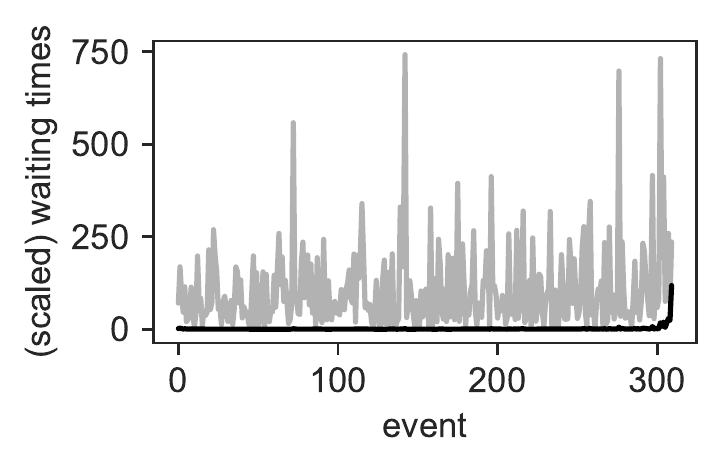}
        \caption{Death}
        \label{fig:plotSIRD(SSA)_interEventTimesD}
    \end{subfigure}
    \caption{Inter event times (dark) and scaled inter event times (bright) for realisation of SSA in Figure~\ref{fig:plotSIRD(SSA)1}. (Connected points to distinguish unscaled and scaled event times.)}
    
    \label{fig:plotSIRD(SSA)_interEventTimes}
\end{figure}

\begin{figure}[ht]
    \centering
    \begin{subfigure}{0.32\linewidth}
        \centering
        \includegraphics[width=\linewidth]{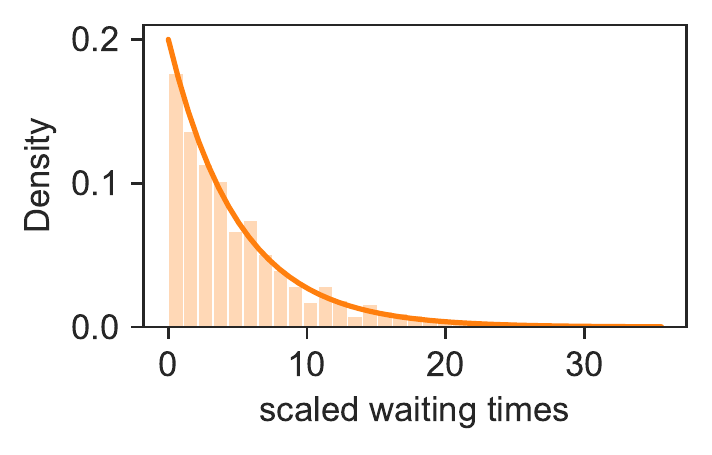}
        \caption{Infection}
        \label{fig:plotSIRD(SSA)_interEventTimesDistributionsI}
    \end{subfigure}
    \hfill
    \begin{subfigure}{0.32\linewidth}
        \centering
        \includegraphics[width=\linewidth]{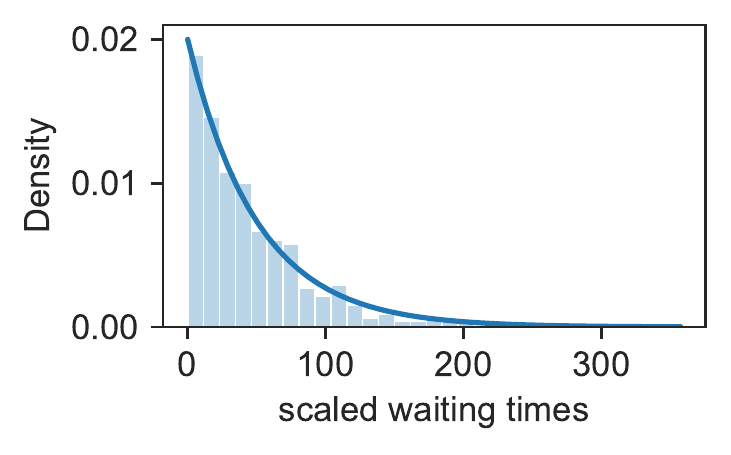}
        \caption{Recovery}
        \label{fig:plotSIRD(SSA)_interEventTimesDistributionsR}
    \end{subfigure}
    \hfill
    \begin{subfigure}{0.32\linewidth}
        \centering
        \includegraphics[width=\linewidth]{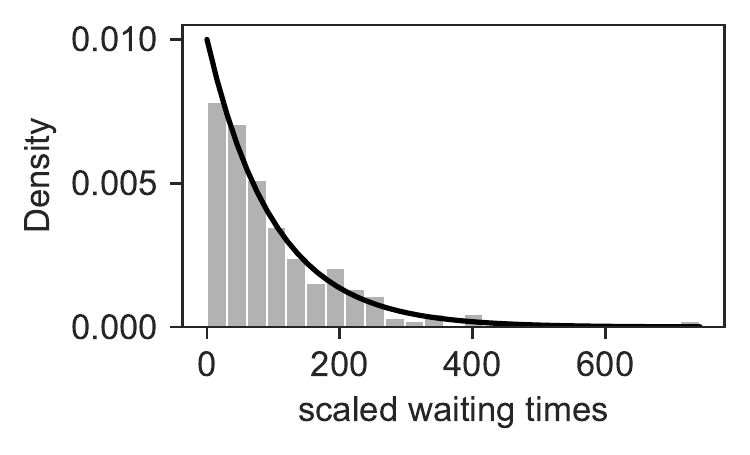}
        \caption{Death}
        \label{fig:plotSIRD(SSA)_interEventTimesDistributionsD}
    \end{subfigure}
    \caption{Densities of scaled inter event times (bars) and theoretical exponential densities (lines) for realisation of SSA in Figure~\ref{fig:plotSIRD(SSA)1}.}
    \label{fig:plotSIRD(SSA)_interEventTimesDistributions}
\end{figure}

\subsection{SIRD with age classes}

A typical use case for Algorithm~\ref{alg:sec:model:subsec:gillespie:subsub:subtypes} is a SIRD model with dividing the types by ages into discrete sub classes, listed in Table~\ref{tab:sec:standard:subsec:sirda:types}. The main rules remain like in \eqref{eq:sec:standard:subsec:sir:rules} but with additional sub rules for the age classes.

\noindent
\begin{minipage}[t]{0.45\textwidth}\vspace{0pt}%
\vspace{-7.5pt}
\captionof{table}{Types of the SIRDage model.}
\label{tab:sec:standard:subsec:sirda:types}
\begin{center}
\extrarowheight=0pt
\addtolength{\extrarowheight}{\belowrulesep}
\aboverulesep=0pt
\belowrulesep=0pt
\begin{tabular}{*{3}{>{\centering\arraybackslash}m{2em}}  l} 
	\toprule
	\multicolumn{3}{c}{Type} & Interpretation  \\
    \hline
	\multicolumn{3}{c}{$\cdot_a$, $a \in A=\{0,1,2\}$} & age sub class:\\
	&&& 0: young,\\
	&&& 1: middle,\\
	\cellcolor{gray!40}$\cdot_0$ & \cellcolor{gray!60}$\cdot_1$ & \cellcolor{gray!80}$\cdot_2$ & 2: old\\
	\hline
	\cellcolor{tabgreen!40}$S_0$ & \cellcolor{tabgreen!60}$S_1$ & \cellcolor{tabgreen!80}$S_2$ & Susceptible  \\ 
	\cellcolor{taborange!40}$I_0$ & \cellcolor{taborange!60}$I_1$ & \cellcolor{taborange!80}$I_2$ & Infectious  \\
	\cellcolor{tabblue!40}$R_0$ & \cellcolor{tabblue!60}$R_1$ & \cellcolor{tabblue!80}$R_2$ & Recovered  \\
	\cellcolor{black!40}\color{white}$D_0$ & \cellcolor{black!60}\color{white}$D_1$ & \cellcolor{black!80}\color{white}$D_2$ & Deceased  \\
	\bottomrule
\end{tabular}
\end{center}

\end{minipage}
\hfill
\begin{minipage}[t]{0.45\textwidth}\vspace{0pt}%
Rules of the SIRDage model:\\ 
for $a, a' \in A$:
\nolinebreak

\begin{equation}
\label{eq:sec:standard:subsec:sirda:rules}
\begin{alignedat}{3}
	S_a + I_{a'} &\xrightarrow{\iota_{a,a'}} && I_a + I_{a'}, &%
&\quad\text{[Infection]}%
	\\
	      I_a    &\xrightarrow{\rho_a}       && R_a,          &%
	   &\quad\text{[Recovery]}%
	   \\
	      I_a    &\xrightarrow{\delta_a}     && D_a,          &%
	   &\quad\text{[Death]}%
	   \\
\end{alignedat}
\end{equation}
\end{minipage}

\vspace{10pt}
The rule parameters are age dependent assuming that older age classes recover less and die more likely, here $\rho = \rho^0 (1.5, 1.0, 0.5)$ and $\delta = \delta^0 (0.25, 1.0, 2.0)$. For infection we simplify that the parameter only depends on the age induced contact rate with contact matrix $C$ based on values in \cite{Wallinga2006}; infectivity and susceptibility are assumed to be the same.
We also model contact reduction by $80 \%$ with the oldest group by scaling the contact matrix: 

\begin{equation*}
    C = \begin{pmatrix}
            76 & 38 &  5\\
            38 & 67 & 13\\
             5 & 13 & 9
        \end{pmatrix},
    \quad %
    \iota^{(fc)} = c \cdot C,%
    \quad%
    \iota^{(rc)} = c \cdot
        \begin{pmatrix} 
            1.0 & 1.0 & 0.2\\
            1.0 & 1.0 & 0.2\\
            0.2 & 0.2 & 1.0
        \end{pmatrix}
        \circ C,
\end{equation*}
where $\circ$ denotes the Hadamard product and $c \approx 0.00485$ is used to scale the infection parameter such that the model with no contact reduction has the same basis reproduction number $\mathcal{R}_0 = 6 \nicefrac{2}{3}$ as the simple SIRD model \eqref{eq:sec:standard:subsec:sir:rules}, see also \cite{rbse}.
The basis reproduction number for the model with contact reduction is slightly less, approximately $6.58$.

\begin{figure}[ht]
    \centering
    \begin{subfigure}{\linewidth}
        \centering
        \includegraphics[width=\linewidth]{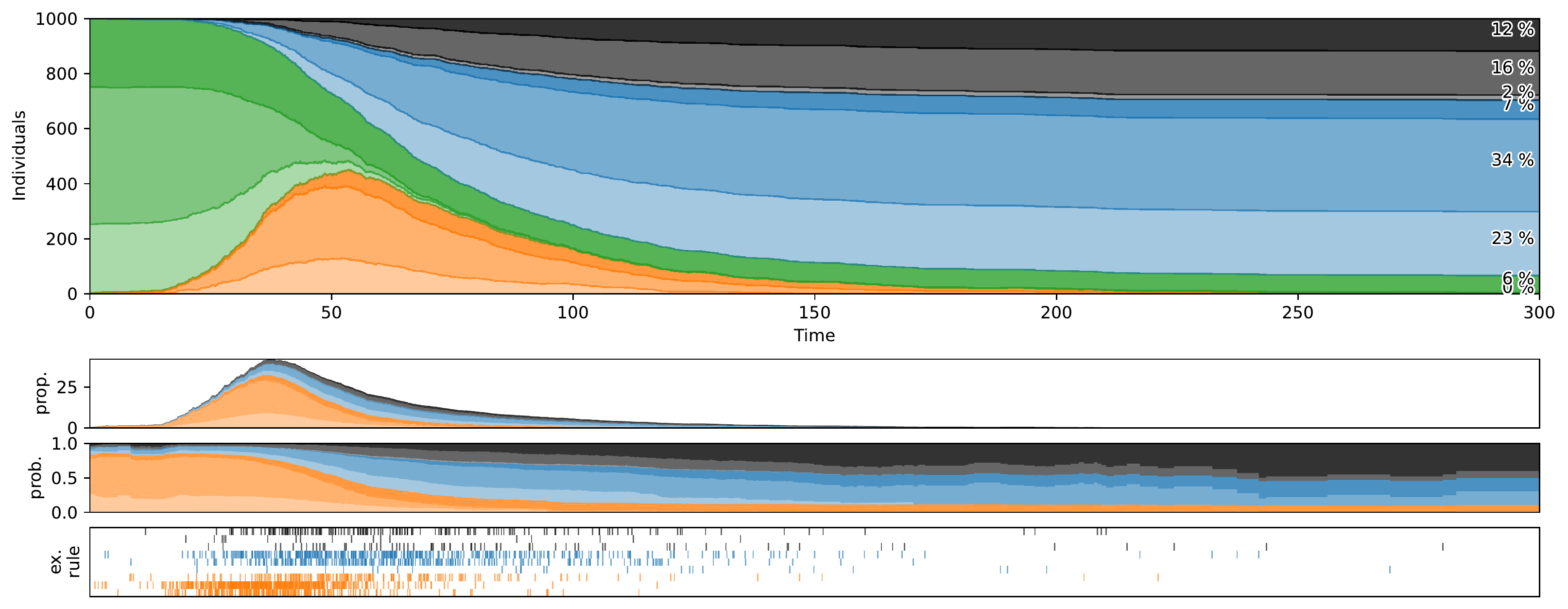}
        \caption{Without contact reduction (fc).}
        \label{fig:plotSIRDa(SSA)_FullContact_extra}
    \end{subfigure}
    
    \begin{subfigure}{\linewidth}
        \centering
        \includegraphics[width=\linewidth]{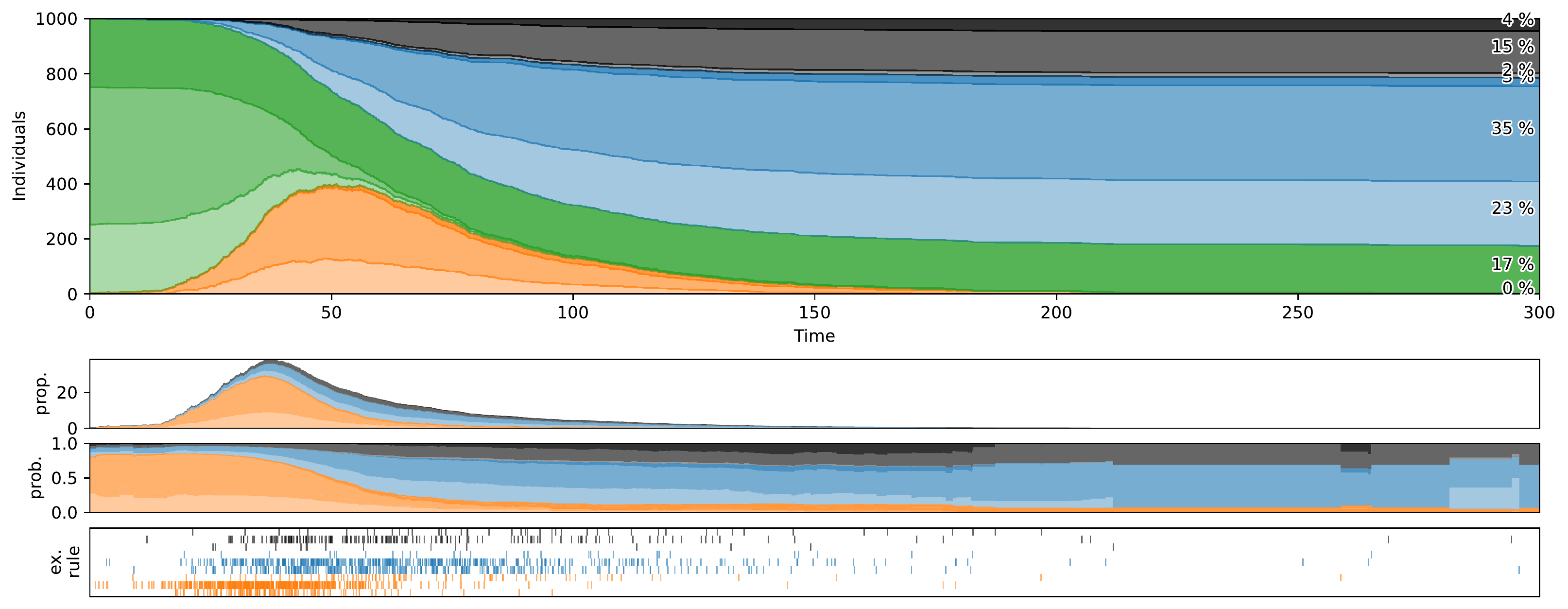}
        \caption{With contact reduction (rc).}
        \label{fig:plotSIRDa(SSA)_ReducedContact_extra}
    \end{subfigure}
    
    \caption{Two realisations of SSA simulation by Algorithm~\ref{alg:sec:model:subsec:gillespie:subsub:discreteupdates2v2} for SIRDage model without / with contact reduction.}
    \label{fig:plotSIRDa(SSA)_extra}
\end{figure}

For both simulation results in Figure~\ref{fig:plotSIRDa(SSA)_extra}, we start with $S(0) = (249, 499, 249)$ susceptibles and $I(0) = (1, 1, 1)$ infectious and choose a random seed of $0$ for a characteristic realisation comparable to ODE simulations.
Figure~\ref{fig:plotSIRDa(SSA)_FullContact_extra} shows that a large fraction of the oldest group suffers severely -- nearly half of the subpopulation dies.
By the contact reduction, simulated in Figure~\ref{fig:plotSIRDa(SSA)_ReducedContact_extra}, a large fraction of the oldest age sub class survives the epidemic without getting infected.
In the plots for propensities and probabilities for clarity only the target subclasses are colour coded by their rule. For the event barcode / rug plot we separate target subclass executed rules by rows.
In the plot of the executed rules we can observe that event densities are quite different: For the youngest group very centered, whereas for the oldest group very rarely and wider. The events for the middle group dominate due to the large size.

\subsection{SIRD with lockdown}

When stochastically simulating more complex models of epidemics, one often has to account for changes based on discrete times. We apply Algorithm~\ref{alg:sec:model:subsec:gillespie:subsub:discreteupdates2v2} for a SIRD model like in \eqref{eq:sec:standard:subsec:sir:rules} but with time dependent infection parameter $\iota(t)$ to model lockdowns of different duration and severity. It is important that $\iota(t)$ is a step function, here

\begin{equation}
\label{eq:SIRDlockdown:infection}
\iota(t) = \iota^0 \cdot
    \begin{cases}
      0.2, & \text{for } \phantom{1}10 \leq t < \phantom{1}20,\\
      0.5, & \text{for } \phantom{1}40 \leq t < \phantom{1}50,\\
      0.1, & \text{for } \phantom{1}50 \leq t < \phantom{1}70,\\
      0.2, & \text{for } 100 \leq t < 120,\\
      0.0, & \text{for } 160 \leq t < 200,\\
      1.0, & \text{else}.
    \end{cases}
\end{equation}

\begin{sidewaysfigure}
    \centering
    \includegraphics[width=\textwidth,]{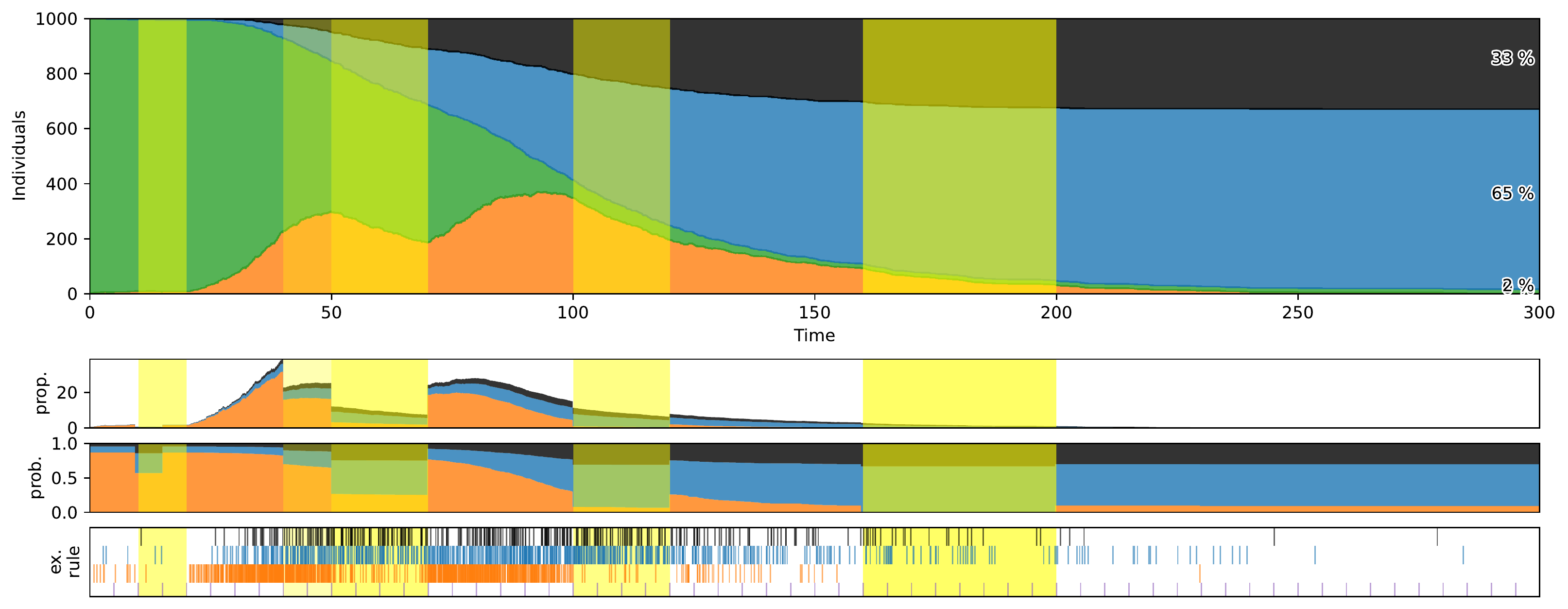}
    
    \caption{One realisation of SSA simulation by Algorithm~\ref{alg:sec:model:subsec:gillespie:subsub:discreteupdates2v2}  with basic parameters, infection parameter step function in \eqref{eq:SIRDlockdown:infection} and seed $= 0$. With propensities, probabilities and executed rules as event barcodes. Lockdown severity indicated by yellow areas.}
    \label{fig:plotSIRD(SSA)lockdown_extra}
\end{sidewaysfigure}

The results in Figure~\ref{fig:plotSIRD(SSA)lockdown_extra} show that the propensities and thus the probabilities depend heavily on lockdowns. The duration and severity are depicted as yellow overlays -- the more opaque the severer.
The reciprocal of the sum of the propensity determines the stochastic time steps between events -- this can also be observed in the density of the bars in the rug plot at the bottom.
The lockdown severity affects exclusively the propensity of the infection rule. Thus changing the proportion in the probabilities highly.
In the extreme case, for $160 \leq t < 200$, where the factor $0$, i.e. severity is $1$, no execution of the infection rule can occur.
Note that lines in the bottom row depict not ticks but the time-steps with $\Delta t = 5$ of the algorithm and show that these global events for potential lockdown changes are produced correctly.
\clearpage

\subsection{SIRD with isolation}

Another case in epidemics, when deterministic and fixed periods of time might appear, is isolation, i.e.,  infected individuals are separated from people who are not sick. We start again from the simple SIRD model in \eqref{eq:sec:standard:subsec:sir:rules} and add a type $Q$ for isolated infected, see Table~\ref{table:sec:standard:subsec:sird-quarantine}.
We assume that with rate parameter $\alpha = 0.05$ the infection of an individual is detected and isolation of $\omega = 100$ days is imposed -- exaggerated for clarity. Furthermore, we assume that the people placed in isolation do not infect anyone or die and are considered recovered at the end of the isolation. Thus only one single independent rule is added in \eqref{eq:sec:standard:subsec:sirdq:rules}, where the dashed arrow denotes a rule with deterministic fixed delay $\omega$.

\noindent
\begin{minipage}[t]{0.45\textwidth}\vspace{0pt}%
\vspace{-7.5pt}
\captionsetup{type=table}
\caption{The types of the standard SIR model with isolation.}
\label{table:sec:standard:subsec:sird-quarantine}
\begin{center}
\extrarowheight=0pt
\addtolength{\extrarowheight}{\belowrulesep}
\aboverulesep=0pt
\belowrulesep=0pt
\begin{tabular}{c   l} 
	\toprule
	Type & Interpretation\\
    \hline
	\cellcolor{tabgreen!80}$S$ & Susceptible\\ 
	\cellcolor{taborange!80}$I$ & Infectious\\
	\cellcolor{tabpink!80}$Q$ & Isolated\\
	\cellcolor{tabblue!80}$R$ & Recovered\\
	\cellcolor{black!80}\color{white}$D$ & Deceased\\
	\bottomrule
\end{tabular}
\end{center}
\end{minipage}
\hfill
\begin{minipage}[t]{0.48\textwidth}\vspace{0pt}%
Rules of the SIRD model with isolation:

\begin{equation}
\label{eq:sec:standard:subsec:sirdq:rules}
\begin{alignedat}{4}
	S + \;&I &&\xrightarrow{\iota}  && I + I, &\quad&\text{[Infection]}\\
          &I &&\xrightarrow{\alpha} && Q \mathop{\dashrightarrow}^{}_{\omega} R,       &\quad&\text{[Isolation \& sub-}\\
          &    &&                     &&  &\quad&\text{\phantom{[}sequent Recovery]}\\
	      &I &&\xrightarrow{\rho}   && R,     &\quad&\text{[Recovery]}\\
          &I &&\xrightarrow{\delta} && D,     &\quad&\text{[Death]}\\
\end{alignedat}
\end{equation}
\end{minipage}
\vspace{10pt}

\begin{figure}[ht]
    \centering
    \includegraphics[width=\linewidth]{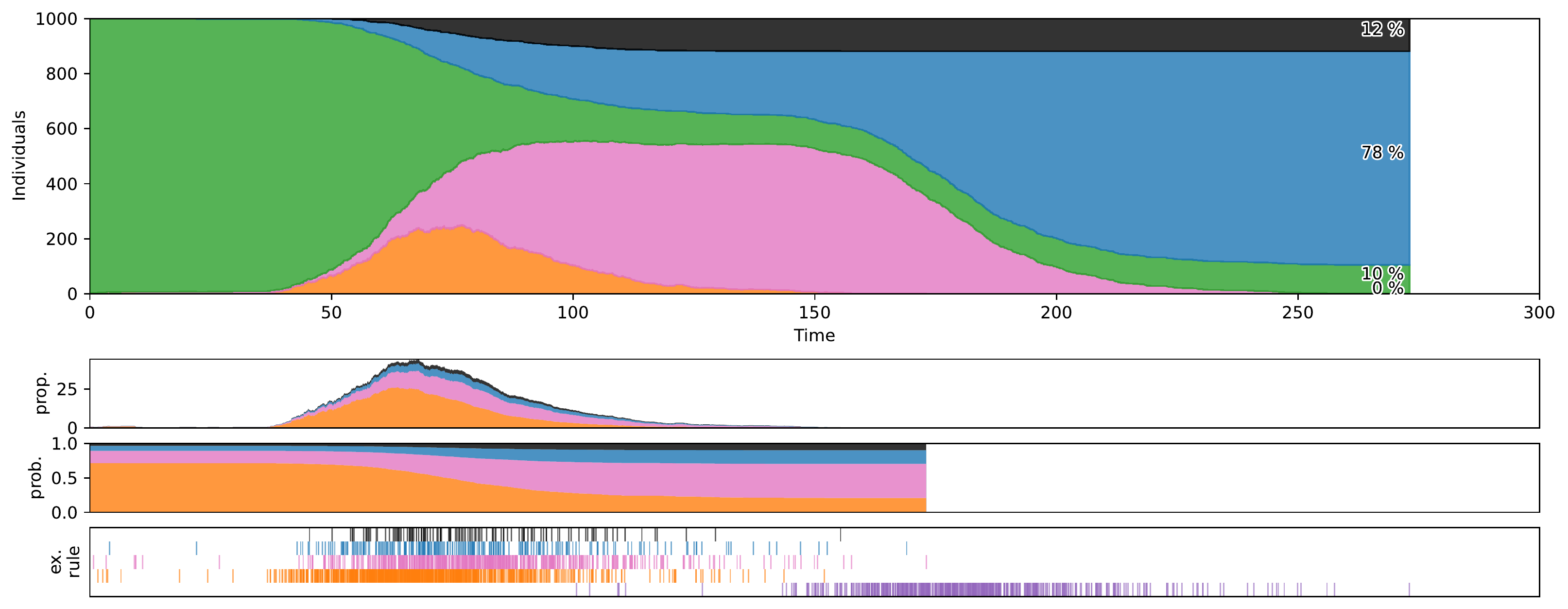}
    \caption{Realisation of SSA simulation by Algorithm~\ref{alg:sec:model:subsec:gillespie:subsub:deterministicDelays}  with basic parameters, $\alpha = 0.05$, $\omega = 100$ and seed $= 0$.}
    \label{fig:plotSIRD(SSA)quarantine_extra}
\end{figure}

The numerical results in Figure~\ref{fig:plotSIRD(SSA)quarantine_extra} show that shortly after the increase of infectious individuals the detection and isolation is the dominant rule, clearly observable in the plot for probabilities.
At time $t \approx 173$ the last infectious individual is detected and placed in isolation.
Hence the propensities for the stochastic rules -- also for the stochastic part of the detection -- are zero, and cause undefined probabilities.
The only events after this time are the release from isolation depicted as purple bars in the bottom row of the event barcode.
Note that the purple row is identical to the pink row for detection and imposing isolation shifted by exactly $\omega = 100$.

\section{Discussion}

In this paper we presented a number of extensions of the Gillespie algorithm for application in epidemiology. However, it is obvious that such generalisations of the algorithm are also useful in other areas of applications. For example, the subtype extension to reaction schemes, generalising \eqref{eq:sec:model:subsec:reaction:reactionscheme} to \eqref{eq:sec:model:subsec:gillespie:subtypes}, can be applied to protein biochemistry, where proteins are considered to be allowed to come in different configurations, i.e. folding states. In pharmacology, outside interventions to the system are the norm, not the exception. This implies the algorithms constructed for interventions like lockdowns have a good use in this area of application as well. The overall tree of generalisations of the Gillespie algorithm is summarised in Figure \ref{fig:Gillespie_algorithms}.

\begin{figure}[ht]
    \centering
    \includegraphics[width=0.5\linewidth]{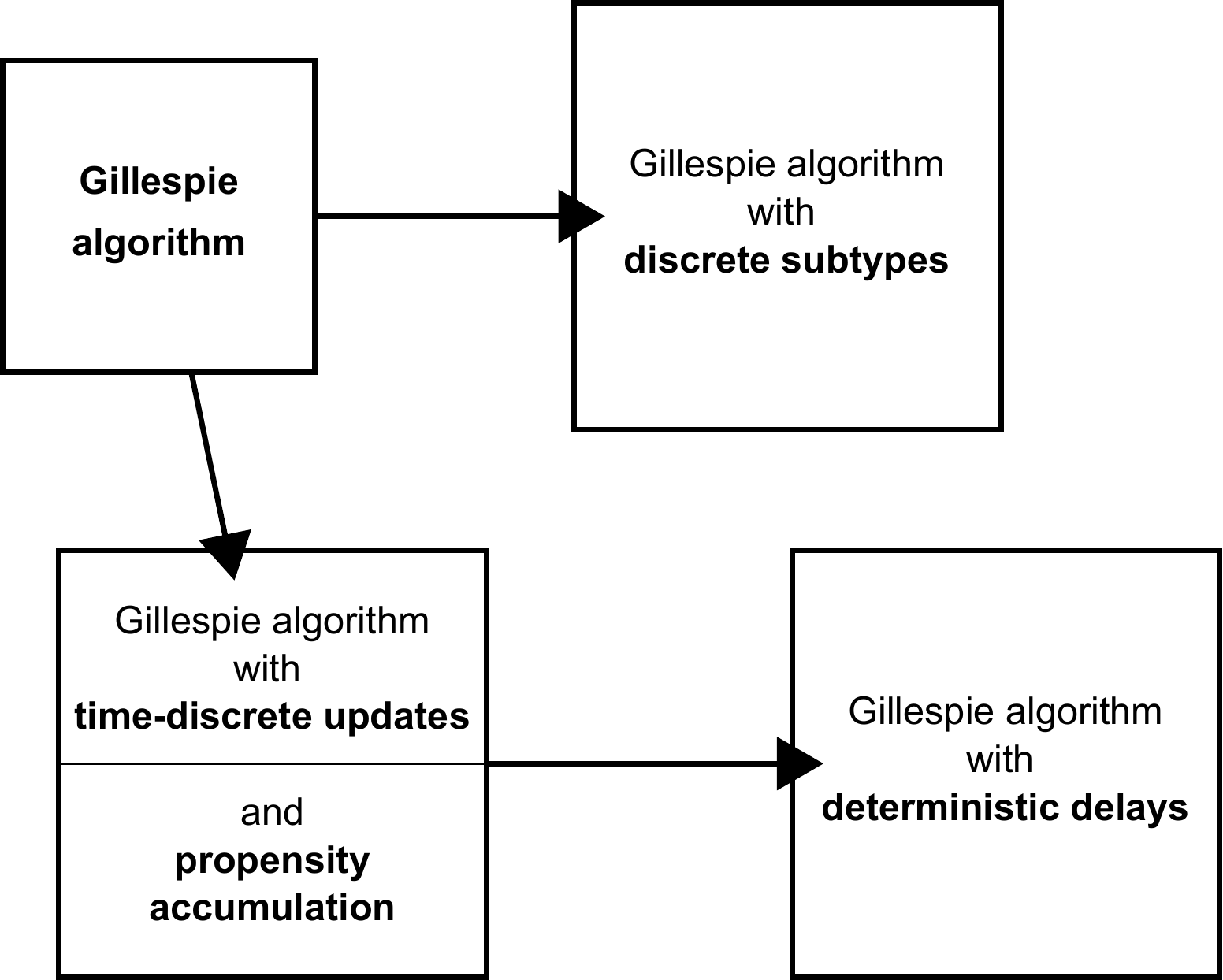}
    \vspace{2mm}
    \caption{Visualisation of the Gillespie algorithms and its generalisations}
    \label{fig:Gillespie_algorithms}
\end{figure}

The examples of epidemiological scenarios shown in this paper have the intention to show how improved expert systems in epidemiology might work. To make forecasts more precise, the suggestion is to compute a multitude of single time series in a hierarchical way. First of all, in contrast to differential equation approaches, we are considering stochastic processes in this paper. This means we have to compute a large set of realisations of the process. This requires more computational power, however it also gives a good knowledge how representative the expectation (associated with the differential equation) of the process is, for example by computing the variance of the multi-realisation ensemble. Then, with the help of the algorithms presented here, interventions can be tested, which again has to be computed for a multi-realisation ensemble. The computational complexity is immense, but unavoidable. The advantage is that predictions can be further evaluated, for example attaching a likelihood to the scenarios. This will give policy makers a better insight about what to expect when they plan certain policies.

\section*{Acknowledgements}
DA acknowledges funding provided by grants  PGC2018-096577-B-I00 and PID2021-127202NB-C21 from MCIN/AEI/10.13039/501100011033 and “ERDF A way of making Europe".
LMK acknowledges support from the the Warwick Research Development Fund through the project ‘Using Partial Differential Equations Techniques to Analyse Data-Rich Phenomena’, the European Union Horizon 2020 research and innovation programmes under the Marie Skłodowska-Curie grant agreement No. 777826 (NoMADS), the Cantab Capital Institute for the Mathematics of Information and Magdalene College, Cambridge (Nevile Research Fellowship).
MK acknowledges the support of all the members of the
\href{https://www.mathematical-modelling.science/index.php/projects/covid-19-epidemiology-project}{Covid-19 Epidemiology Project}.
Special thanks to Joan Saldaña for contributing references
and ideas.





\bibliography{references}


\begin{thebibliography}{26}
\ifx \bisbn   \undefined \def \bisbn  #1{ISBN #1}\fi
\ifx \binits  \undefined \def \binits#1{#1}\fi
\ifx \bauthor  \undefined \def \bauthor#1{#1}\fi
\ifx \batitle  \undefined \def \batitle#1{#1}\fi
\ifx \bjtitle  \undefined \def \bjtitle#1{#1}\fi
\ifx \bvolume  \undefined \def \bvolume#1{\textbf{#1}}\fi
\ifx \byear  \undefined \def \byear#1{#1}\fi
\ifx \bissue  \undefined \def \bissue#1{#1}\fi
\ifx \bfpage  \undefined \def \bfpage#1{#1}\fi
\ifx \blpage  \undefined \def \blpage #1{#1}\fi
\ifx \burl  \undefined \def \burl#1{\textsf{#1}}\fi
\ifx \doiurl  \undefined \def \doiurl#1{\url{https://doi.org/#1}}\fi
\ifx \betal  \undefined \def \betal{\textit{et al.}}\fi
\ifx \binstitute  \undefined \def \binstitute#1{#1}\fi
\ifx \binstitutionaled  \undefined \def \binstitutionaled#1{#1}\fi
\ifx \bctitle  \undefined \def \bctitle#1{#1}\fi
\ifx \beditor  \undefined \def \beditor#1{#1}\fi
\ifx \bpublisher  \undefined \def \bpublisher#1{#1}\fi
\ifx \bbtitle  \undefined \def \bbtitle#1{#1}\fi
\ifx \bedition  \undefined \def \bedition#1{#1}\fi
\ifx \bseriesno  \undefined \def \bseriesno#1{#1}\fi
\ifx \blocation  \undefined \def \blocation#1{#1}\fi
\ifx \bsertitle  \undefined \def \bsertitle#1{#1}\fi
\ifx \bsnm \undefined \def \bsnm#1{#1}\fi
\ifx \bsuffix \undefined \def \bsuffix#1{#1}\fi
\ifx \bparticle \undefined \def \bparticle#1{#1}\fi
\ifx \barticle \undefined \def \barticle#1{#1}\fi
\bibcommenthead
\ifx \bconfdate \undefined \def \bconfdate #1{#1}\fi
\ifx \botherref \undefined \def \botherref #1{#1}\fi
\ifx \url \undefined \def \url#1{\textsf{#1}}\fi
\ifx \bchapter \undefined \def \bchapter#1{#1}\fi
\ifx \bbook \undefined \def \bbook#1{#1}\fi
\ifx \bcomment \undefined \def \bcomment#1{#1}\fi
\ifx \oauthor \undefined \def \oauthor#1{#1}\fi
\ifx \citeauthoryear \undefined \def \citeauthoryear#1{#1}\fi
\ifx \endbibitem  \undefined \def \endbibitem {}\fi
\ifx \bconflocation  \undefined \def \bconflocation#1{#1}\fi
\ifx \arxivurl  \undefined \def \arxivurl#1{\textsf{#1}}\fi
\csname PreBibitemsHook\endcsname

\bibitem{rbse}
\begin{barticle}
\bauthor{\bsnm{Alonso}, \binits{D.}},
\bauthor{\bsnm{Bauer}, \binits{S.}},
\bauthor{\bsnm{Kirkilionis}, \binits{M.}},
\bauthor{\bsnm{Kreusser}, \binits{L.M.}},
\bauthor{\bsnm{Sbano}, \binits{L.}}:
\batitle{A rule-based epidemiological framework for modelling and simulation in
  the context of the covid-19 pandemic}.
\bjtitle{arXiv}
(\byear{2021})
{\href{https://arxiv.org/abs/2111.07336}{{2111.07336}}}.
\doiurl{10.48550/arXiv.2111.07336}
\end{barticle}
\endbibitem

\bibitem{Manrubia2020}
\begin{barticle}
\bauthor{\bsnm{Manrubia}, \binits{S.}}:
\batitle{{The Uncertain Future in How a Virus Spreads}}.
\bjtitle{Physics}
\bvolume{13},
\bfpage{1}--\blpage{3}
(\byear{2020}).
\doiurl{10.1103/physics.13.166}
\end{barticle}
\endbibitem

\bibitem{Castro26190}
\begin{barticle}
\bauthor{\bsnm{Castro}, \binits{M.}},
\bauthor{\bsnm{Ares}, \binits{S.}},
\bauthor{\bsnm{Cuesta}, \binits{J.A.}},
\bauthor{\bsnm{Manrubia}, \binits{S.}}:
\batitle{{The turning point and end of an expanding epidemic cannot be
  precisely forecast}}.
\bjtitle{Proceedings of the National Academy of Sciences}
\bvolume{117}(\bissue{42}),
\bfpage{26190}--\blpage{26196}
(\byear{2020}).
\doiurl{10.1073/pnas.2007868117}
\end{barticle}
\endbibitem

\bibitem{Lara-Tuprio:2022aa}
\begin{botherref}
\oauthor{\bsnm{Lara-Tuprio}, \binits{E.d.}},
\oauthor{\bsnm{Estuar}, \binits{M.}},
\oauthor{\bsnm{Sescon{\ldots}}, \binits{J.}}:
Economic losses from covid-19 cases in the philippines: a dynamic model of
  health and economic policy trade-offs.
Humanities and Social {\ldots}
(2022)
\end{botherref}
\endbibitem

\bibitem{Liu:2022aa}
\begin{botherref}
\oauthor{\bsnm{Liu}, \binits{Z.}},
\oauthor{\bsnm{Deng}, \binits{Z.}},
\oauthor{\bsnm{Zhu}, \binits{B.}},
\oauthor{\bsnm{Ciais}, \binits{P.}},
\oauthor{\bsnm{Davis}, \binits{S.}},
\oauthor{\bsnm{Tan{\ldots}}, \binits{J.}}:
Global patterns of daily co2 emissions reductions in the first year of
  covid-19.
Nature {\ldots}
(2022)
\end{botherref}
\endbibitem

\bibitem{Wang:2022aa}
\begin{botherref}
\oauthor{\bsnm{Wang}, \binits{H.}},
\oauthor{\bsnm{Zheng}, \binits{Y.}},
\oauthor{\bsnm{Jonge}, \binits{M.d.}},
\oauthor{\bsnm{Wang{\ldots}}, \binits{R.}}:
Lockdown measures during the covid-19 pandemic strongly impacted the
  circulation of respiratory pathogens in southern china.
Scientific Reports
(2022)
\end{botherref}
\endbibitem

\bibitem{Shan:2021aa}
\begin{botherref}
\oauthor{\bsnm{Shan}, \binits{Y.}},
\oauthor{\bsnm{Ou}, \binits{J.}},
\oauthor{\bsnm{Wang}, \binits{D.}},
\oauthor{\bsnm{Zeng}, \binits{Z.}},
\oauthor{\bsnm{Zhang{\ldots}}, \binits{S.}}:
Impacts of covid-19 and fiscal stimuli on global emissions and the paris
  agreement.
Nature Climate {\ldots}
(2021)
\end{botherref}
\endbibitem

\bibitem{Wells:2022aa}
\begin{botherref}
\oauthor{\bsnm{Wells}, \binits{C.}},
\oauthor{\bsnm{Pandey}, \binits{A.}},
\oauthor{\bsnm{Moghadas{\ldots}}, \binits{S.}}:
Comparative analyses of eighteen rapid antigen tests and rt-pcr for covid-19
  quarantine and surveillance-based isolation.
Communications {\ldots}
(2022)
\end{botherref}
\endbibitem

\bibitem{Aleta:2020aa}
\begin{barticle}
\bauthor{\bsnm{Aleta}, \binits{A.}},
\bauthor{\bsnm{Martin-Corral}, \binits{D.}},
\bauthor{\bparticle{Pastore~y} \bsnm{Piontti}, \binits{A.}},
\bauthor{\bsnm{Ajelli}, \binits{M.}},
\bauthor{\bsnm{Litvinova}, \binits{M.}},
\bauthor{\bsnm{Chinazzi}, \binits{M.}},
\bauthor{\bsnm{Dean}, \binits{N.E.}},
\bauthor{\bsnm{Halloran}, \binits{M.E.}},
\bauthor{\bsnm{Longini~Jr}, \binits{I.M.}},
\bauthor{\bsnm{Merler}, \binits{S.}}:
\batitle{Modelling the impact of testing, contact tracing and household
  quarantine on second waves of covid-19}.
\bjtitle{Nature Human Behaviour}
\bvolume{4}(\bissue{9}),
\bfpage{964}--\blpage{971}
(\byear{2020})
\end{barticle}
\endbibitem

\bibitem{Eilersen:2020aa}
\begin{barticle}
\bauthor{\bsnm{Eilersen}, \binits{A.}},
\bauthor{\bsnm{Sneppen}, \binits{K.}}:
\batitle{Cost-benefit of limited isolation and testing in covid-19 mitigation.}
\bjtitle{Sci Rep}
\bvolume{10}(\bissue{1}),
\bfpage{18543}
(\byear{2020})
\end{barticle}
\endbibitem

\bibitem{Block:2020aa}
\begin{barticle}
\bauthor{\bsnm{Block}, \binits{P.}},
\bauthor{\bsnm{Hoffman}, \binits{M.}},
\bauthor{\bsnm{Raabe}, \binits{I.J.}},
\bauthor{\bsnm{Dowd}, \binits{J.B.}},
\bauthor{\bsnm{Rahal}, \binits{C.}},
\bauthor{\bsnm{Kashyap}, \binits{R.}},
\bauthor{\bsnm{Mills}, \binits{M.C.}}:
\batitle{Social network-based distancing strategies to flatten the covid-19
  curve in a post-lockdown world}.
\bjtitle{Nature Human Behaviour}
\bvolume{4}(\bissue{6}),
\bfpage{588}--\blpage{596}
(\byear{2020})
\end{barticle}
\endbibitem

\bibitem{Hall:2022aa}
\begin{botherref}
\oauthor{\bsnm{Hall}, \binits{V.}},
\oauthor{\bsnm{Ferreira}, \binits{V.}},
\oauthor{\bsnm{Wood}, \binits{H.}},
\oauthor{\bsnm{Ierullo{\ldots}}, \binits{M.}}:
Delayed-interval bnt162b2 mrna covid-19 vaccination enhances humoral immunity
  and induces robust t cell responses.
Nature {\ldots}
(2022)
\end{botherref}
\endbibitem

\bibitem{Waxman:2022aa}
\begin{barticle}
\bauthor{\bsnm{Waxman}, \binits{J.G.}},
\bauthor{\bsnm{Makov-Assif}, \binits{M.}},
\bauthor{\bsnm{Reis}, \binits{B.Y.}},
\bauthor{\bsnm{Netzer}, \binits{D.}},
\bauthor{\bsnm{Balicer}, \binits{R.D.}},
\bauthor{\bsnm{Dagan}, \binits{N.}},
\bauthor{\bsnm{Barda}, \binits{N.}}:
\batitle{Comparing covid-19-related hospitalization rates among individuals
  with infection-induced and vaccine-induced immunity in israel}.
\bjtitle{Nature communications}
\bvolume{13}(\bissue{1}),
\bfpage{1}--\blpage{6}
(\byear{2022})
\end{barticle}
\endbibitem

\bibitem{Gillespie1976}
\begin{barticle}
\bauthor{\bsnm{Gillespie}, \binits{D.T.}}:
\batitle{A general method for numerically simulating the stochastic time
  evolution of coupled chemical reactions}.
\bjtitle{Journal of Computational Physics}
\bvolume{22}(\bissue{4}),
\bfpage{403}--\blpage{434}
(\byear{1976}).
\doiurl{10.1016/0021-9991(76)90041-3}
\end{barticle}
\endbibitem

\bibitem{Gillespie1992}
\begin{barticle}
\bauthor{\bsnm{Gillespie}, \binits{D.T.}}:
\batitle{A rigorous derivation of the chemical master equation}.
\bjtitle{Physica A: Statistical Mechanics and its Applications}
\bvolume{188}(\bissue{1}),
\bfpage{404}--\blpage{425}
(\byear{1992}).
\doiurl{10.1016/0378-4371(92)90283-V}
\end{barticle}
\endbibitem

\bibitem{Gillespie2000}
\begin{barticle}
\bauthor{\bsnm{Gillespie}, \binits{D.T.}}:
\batitle{The chemical langevin equation}.
\bjtitle{The Journal of Chemical Physics}
\bvolume{113}(\bissue{1}),
\bfpage{297}--\blpage{306}
(\byear{2000}).
\doiurl{10.1063/1.481811}
\end{barticle}
\endbibitem

\bibitem{Gillespie_ME_with_waiting_times}
\begin{barticle}
\bauthor{\bsnm{Gillespie}, \binits{D.T.}}:
\batitle{Master equations for random walks with arbitrary pausing time
  distributions}.
\bjtitle{Physics Letters A}
\bvolume{64}(\bissue{1}),
\bfpage{22}--\blpage{24}
(\byear{1977}).
\doiurl{10.1016/0375-9601(77)90513-8}
\end{barticle}
\endbibitem

\bibitem{Gibson2000}
\begin{barticle}
\bauthor{\bsnm{Gibson}, \binits{M.A.}},
\bauthor{\bsnm{Bruck}, \binits{J.}}:
\batitle{{Efficient Exact Stochastic Simulation of Chemical Systems with Many
  Species and Many Channels}}.
\bjtitle{The Journal of Physical Chemistry A}
\bvolume{104}(\bissue{9}),
\bfpage{1876}--\blpage{1889}
(\byear{2000}).
\doiurl{10.1021/jp993732q}
\end{barticle}
\endbibitem

\bibitem{Bernstein2005}
\begin{barticle}
\bauthor{\bsnm{Bernstein}, \binits{D.}}:
\batitle{{Simulating mesoscopic reaction-diffusion systems using the Gillespie
  algorithm}}.
\bjtitle{Physical Review E - Statistical, Nonlinear, and Soft Matter Physics}
\bvolume{71}(\bissue{4}),
\bfpage{1}--\blpage{13}
(\byear{2005}).
\doiurl{10.1103/PhysRevE.71.041103}
\end{barticle}
\endbibitem

\bibitem{Masuda}
\begin{barticle}
\bauthor{\bsnm{Masuda}, \binits{N.}},
\bauthor{\bsnm{Rocha}, \binits{L.E.C.}}:
\batitle{A gillespie algorithm for non-markovian stochastic processes}.
\bjtitle{SIAM Review}
\bvolume{60}(\bissue{1}),
\bfpage{95}--\blpage{115}
(\byear{2018})
\end{barticle}
\endbibitem

\bibitem{Bog}
\begin{barticle}
\bauthor{\bsnm{Bogu\~n\'a}, \binits{M.}},
\bauthor{\bsnm{Lafuerza}, \binits{L.F.}},
\bauthor{\bsnm{Toral}, \binits{R.}},
\bauthor{\bsnm{Serrano}, \binits{M.A.}}:
\batitle{Simulating non-markovian stochastic processes}.
\bjtitle{Phys. Rev. E}
\bvolume{90},
\bfpage{042108}
(\byear{2014}).
\doiurl{10.1103/PhysRevE.90.042108}
\end{barticle}
\endbibitem

\bibitem{CCTRW}
\begin{barticle}
\bauthor{\bsnm{Aquino}, \binits{T.}},
\bauthor{\bsnm{Dentz}, \binits{M.}}:
\batitle{Chemical continuous time random walks}.
\bjtitle{Phys. Rev. Lett.}
\bvolume{119},
\bfpage{230601}
(\byear{2017}).
\doiurl{10.1103/PhysRevLett.119.230601}
\end{barticle}
\endbibitem

\bibitem{Masuda2}
\begin{barticle}
\bauthor{\bsnm{Masuda}, \binits{N.}},
\bauthor{\bsnm{Vestergaard}, \binits{C.L.}}:
\batitle{Gillespie algorithms for stochastic multiagent dynamics in populations
  and networks}.
\bjtitle{arXiv}
(\byear{2021}).
\doiurl{10.48550/arXiv.2112.05293}
\end{barticle}
\endbibitem

\bibitem{Gardiner}
\begin{bbook}
\bauthor{\bsnm{Gardiner}, \binits{C.}}:
\bbtitle{Stochastic Methods: A Handbook for the Natural and Social Sciences}.
\bsertitle{Springer Series in Synergetics},
vol. \bseriesno{3}
(\byear{2009})
\end{bbook}
\endbibitem

\bibitem{Anderson2007}
\begin{barticle}
\bauthor{\bsnm{Anderson}, \binits{D.F.}}:
\batitle{A modified next reaction method for simulating chemical systems with
  time dependent propensities and delays}.
\bjtitle{The Journal of chemical physics}
\bvolume{127}(\bissue{21}),
\bfpage{214107}
(\byear{2007})
\end{barticle}
\endbibitem

\bibitem{Wallinga2006}
\begin{barticle}
\bauthor{\bsnm{Wallinga}, \binits{J.}},
\bauthor{\bsnm{Teunis}, \binits{P.}},
\bauthor{\bsnm{Kretzschmar}, \binits{M.}}:
\batitle{Using data on social contacts to estimate age-specific transmission
  parameters for respiratory-spread infectious agents}.
\bjtitle{American journal of epidemiology}
\bvolume{164}(\bissue{10}),
\bfpage{936}--\blpage{944}
(\byear{2006})
\end{barticle}
\endbibitem

\end{thebibliography}


\end{document}